\documentclass[a4paper,twoside]{IEEEtran}

\usepackage[utf8]{inputenc} 
\usepackage[T1]{fontenc}
\usepackage{url}
\usepackage{cite}
\usepackage[cmex10]{amsmath} 
\usepackage{amsfonts,amssymb,amsthm,bm,bbm}
\usepackage{graphicx,subcaption}
\usepackage{xcolor}
\usepackage{enumerate}
\usepackage{algpseudocode}
\usepackage{cite}

\newcommand{\eqdef}{:=}
\newcommand{\reqdef}{=:}
\newcommand{\E}{\mathsf{E}}		
\newcommand{\V}{\mathsf{V}}			
\newcommand{\stdset}[1]{\mathbb{#1}}	
\newcommand{\set}[1]{\mathcal{#1}}		
\renewcommand{\vec}[1]{\bm{#1}}		
\newcommand{\CN}{\mathcal{CN}}			
\newcommand{\herm}{\mathsf{H}}			
\newcommand{\T}{\mathsf{T}}				

\newcommand{\esssup}[1]{\underset{#1}{\mathrm{ess~sup}~}}
\newcommand{\essinf}[1]{\underset{#1}{\mathrm{ess~inf}~}}

\newtheorem{lemma}{Lemma}
\newtheorem{definition}{Definition}

\newtheorem*{theorem}{Main result}
\newtheorem{proposition}{Proposition}

\newtheorem{remark}{Remark}
\newtheorem{corollary}{Corollary}

\begin{document}
\title{A joint channel estimation and beamforming separation principle for massive MIMO systems} 

\author{Lorenzo Miretti,~\IEEEmembership{Member,~IEEE}, Slawomir Sta\'nczak,~\IEEEmembership{Senior Member,~IEEE}, and 
Giuseppe Caire,~\IEEEmembership{Fellow,~IEEE}
\thanks{During the preparation of this article, Lorenzo Miretti was with the Technische Universität Berlin, 10587
Berlin, Germany, and also with the Fraunhofer Institute for Telecommunications Heinrich-Hertz-Institut (HHI), 10587 Berlin, Germany. 
He is now with Ericsson Research Germany.}
\thanks{Slawomir Stanczak is with the Technische Universität Berlin, 10587
Berlin, Germany, and also with the Fraunhofer Institute for Telecommunications Heinrich-Hertz-Institut (HHI), 10587 Berlin, Germany
}
\thanks{Giuseppe Caire is with the Technische Universität Berlin, 10587 Berlin, Germany
}
}

\maketitle
\IEEEpubid{\begin{minipage}{\textwidth}\ \\[12pt] \centering
  This work has been submitted to the IEEE for possible publication.\\ Copyright may be transferred without notice, after which this version may no longer be accessible.
\end{minipage}}


\begin{abstract}
\IEEEpubidadjcol
We demonstrate that separating beamforming (i.e., downlink precoding and uplink combining) and channel estimation in multi-user MIMO wireless systems incurs no loss of optimality under general conditions that apply to a wide variety of models in the literature, including canonical reciprocity-based cellular and cell-free massive MIMO system models. Specifically, we provide conditions under which optimal processing in terms of  ergodic achievable rates can be decomposed into minimum mean-square error (MMSE) channel estimation followed by MMSE beamforming, for both centralized and distributed architectures. Applications of our results are illustrated in terms of concrete examples and numerical simulations.
\end{abstract}

\section{Introduction}
Since its inception in Marzetta's seminal work \cite{marzetta2010noncooperative}, the massive MIMO paradigm has emerged as the \textit{de facto} approach for realizing the  capacity gains of wireless multi-user MIMO systems promised by the classical communication and information theoretic literature \cite{gesbert2007shifting,caire2003broadcast,weingarten2006capacity}. The central innovations underpinning its success are: (i) the use of an excessive number of infrastructure antennas compared to the number of jointly transmitted data streams within the same time-frequency resource; and (ii) the use of reciprocity-based channel estimation in time-division duplexing mode \cite{marzetta2015intro}. These innovations enable significant spatial multiplexing gains through practical signal processing algorithms, and are now routinely incorporated (in some form) into the high-end radio portfolios of all major network vendors. From a theoretical perspective, these innovations have been accompanied by the development of an elegant framework for the study of many crucial system aspects such as the impact of imperfect channel estimation \cite{yin2014dealing,bjornson2016howmany,bjornson2018unlimited}. Due to its generality, this framework has also been successfully applied to distributed variants of the original massive MIMO paradigm, such as the so-called \textit{cell-free} massive MIMO paradigm \cite{ngo2017cell}, which is the most recent version of the \textit{network} MIMO paradigm \cite{gesbert2010multicell}. At present, massive MIMO theory has reached a remarkable maturity, with many textbooks providing accessible foundations and powerful design guidelines \cite{marzetta2016fundamentals,massivemimobook,heath2018foundations,demir2021}. However, the current theory still presents important limitations and open questions. 

\subsection{Joint channel estimation and beamforming}
\IEEEpubidadjcol
One potential limitation, which is the focus of this study, is that pilot-based channel estimation and beamforming design are performed in a somewhat disjoint fashion, i.e., as separate optimization problems. Specifically, the current literature imposes processing architectures that separate the channel estimation and beamforming functions \textit{a priori}. While this approach has the obvious advantage of simplicity (which we certainly do not dismiss), it has the drawback of neglecting the potential benefit of a joint design of the two processing functions, or better yet, of optimal beamformers that take the received raw pilot signals directly as input. Interestingly, the joint approach aligns well with the recent trend towards the aggregation of different physical layer functions \cite{hoydis2021toward}, enabled by the rapid emergence of modern artificial intelligence techniques capable of handling seemingly intractable joint optimization problems. However, to justify the corresponding research and development effort, it is necessary to advance the current massive MIMO theoretical framework to quantify the potential performance gains (if any) of joint channel estimation and beamforming.

\subsection{Main contribution}
In this study, we contribute to closing the above gap in the literature by providing general conditions under which separating the channel estimation and beamforming functions incurs no loss of optimality. Furthermore, we show that when these conditions are met, the two functions can be individually optimized according to a minimum mean-squared error (MMSE) criterion. Informally, we derive the following joint channel estimation and beamforming separation principle.
\begin{theorem}[Informal]
In canonical massive MIMO system models, jointly optimal channel estimation and beamforming in terms of ergodic achievable rates can be decomposed into MMSE channel estimation followed by MMSE beamforming.
\end{theorem}
More precisely, the above separation principle applies to all multi-cell massive MIMO models in the literature that assume standard pilot-based, Gaussian-distributed measurements \cite[Chapter~3]{massivemimobook} of spatially correlated Gaussian fading channels \cite[Chapter~2]{massivemimobook}, and that measure performance in terms of  ergodic capacity lower bounds \cite[Chapter~4]{massivemimobook} \cite{caire2018ergodic} based on treating interference as noise, linear array processing, and either the popular \textit{use-and-then-forget} (UatF) / \textit{hardening} technique, or the more conventional \textit{coherent decoding} technique. The separation principle also directly applies to all canonical cell-free massive MIMO models based on centralized processing architectures with user-centric clustering \cite{demir2021}. Moreover, when restricted to the UatF/hardening bounding technique, it extends to fairly general distributed processing architectures leveraging local channel measurements at each processing unit and potentially an information-sharing procedure that preserves the Gaussianity of the measurements. Notably, in our terminology, \textit{distributed} processing encompasses purely local processing architectures, as in \cite{ngo2017cell} and \cite{demir2021}, as well as more involved architectures based on sequential \cite{miretti2021team} or delayed \cite{miretti2024delayed} information sharing. An extension to the coherent decoding lower bound for the distributed processing case is also possible, provided that the channel state information (CSI) used for decoding is also used as common information for beamforming. Table~\ref{tab:comparison} summarizes the main cell-free massive MIMO models for which our separation principle applies.
\begin{table}[htbp]
    \centering
        \caption{Summary of beamforming architectures for which the channel estimation - beamforming separation principle applies.}
    \begin{tabular}{|l|c|c|c|c|}
        \hline
        & Centralized & Local & Common CSI & General\\
        \hline
        UatF/hard. & Yes & Yes & Yes & Yes\\
        Coh. dec. & Yes & No & Yes & No \\
        \hline
    \end{tabular}
    \label{tab:comparison}
\end{table}
\vspace{-0.7cm}
\subsection{Organization of the paper}
In Section~\ref{sec:model}, we first provide a unified uplink and downlink system model that encompasses a wide variety of models from the cellular and cell-free massive MIMO literature. Focusing on the uplink, in Section~\ref{sec:uplink}, we discuss how the joint channel estimation and beamforming optimization problem, formulated as an ergodic rate maximization problem, can be cast as an equivalent generalized MMSE problem. Building on the results in Section~\ref{sec:uplink}, and still focusing on the uplink case, in Section~\ref{sec:main}, we establish the main results of our study, i.e., the various forms of the joint channel estimation and beamforming separation principle. Furthermore, in Section~\ref{sec:duality}, we provide conditions under which downlink ergodic rate maximization problems admit a solution over a dual uplink channel, and thus under which the main results in Section~\ref{sec:main} extend to the downlink case. Finally, in Section~\ref{sec:applications} and Section~\ref{sec:simulations}, we illustrate our results with concrete examples of applications. 

\textit{Notation}: We denote by $\stdset{R}_+$ and $\stdset{R}_{++}$ the sets of, respectively, nonnegative and positive reals. Lower and upper case bold letters are used for vectors and matrices respectively, while calligraphic letters are used for sets. Let $(\Omega,\Sigma,\mathbb{P})$ be a probability space. We denote by $\set{C}$ the set of complex-valued random variables, i.e., of measurable functions $\Omega \to \stdset{C}$. Any random quantity $A$ is assumed to belong to a product set $\set{C}^{n\times m}$ for some $n,m \in \stdset{N}$, i.e., to take value in $\set{A}\eqdef\stdset{C}^{n\times m}$ with entries organized as vectors, matrices, or tuples. Given a random quantity $A$, we denote by $\set{F}_A\subseteq \set{C}$ the subset of complex-valued functions of $A$, i.e., of measurable functions $\set{A} \to \stdset{C}$. Note that $F\in \set{F}_A$ is a random variable. Similarly, we define the set $\set{P}$ of nonnegative random variables, i.e., of measurable functions $\Omega \to \stdset{R}_{+}$ and its subset $\set{P}_A\subseteq \set{P}$ of non-negative functions of $A$. The expected value, (element-wise) variance, and differential entropy of a random quantity $A$ are denoted by $\E[A]$, $\V(A)$, and $h(A)$ respectively, and their conditional version given a random quantity $B$ by $\E[A|B]$, $\V(A|B)\eqdef \E[|A|^2|B]-|\E[A|B]|^2$, and $h(A|B)$. All (in)equalities involving random variables should be intended as almost sure (in)equalities within equivalence classes of random variables. Given a (possibly uncountable) set of $\stdset{R}$-valued random variables $\set{B}$,  its \textit{essential} supremum (denoted by $\mathrm{ess~sup}~\set{B}$) is defined as the unique $\stdset{R}$-valued random variable $B$ satisfying 
\begin{enumerate}[(i)]
\item $(\forall A\in \set{B})~B\geq A$,
\item For all $B'$ s.t. $(\forall A\in \set{B})~B'\geq A$, we have $B'\geq B$.
\end{enumerate}
The essential infimum ($\mathrm{ess~inf}~\set{B}$) is similarly defined by reversing all inequalities. Note that the above definition, common in the stochastic processes and control literature (see \cite[Definition~2.5]{barron2003conditional} and references therein), differs from the standard definition in the functional analysis literature \cite{luenberger1997optimization}, which applies to a single random variable, and not to a set. 


\section{System model}\label{sec:model}
We consider a wireless multi-user MIMO communication network composed by $K$ single-antenna users indexed by $\set{K}= \{1,\ldots,K\}$, and a potentially large number $M$ of infrastructure antennas indexed by $\set{M}= \{1,\ldots,M\}$, which can be colocated as in common single-cell massive MIMO system models, or geographically distributed as in more advanced multi-cell and cell-free massive MIMO system models. For both uplink and downlink, we consider transmission over a canonical Gaussian vector channel with fading governed by an arbitrarily distributed channel matrix $\vec{H} = \begin{bmatrix}\vec{h}_1&\ldots & \vec{h}_K \end{bmatrix}\in \set{C}^{M\times K}$, using conventional single-user coding schemes which treat interference as noise, and a linear combining/precoding (beamforming) matrix $\vec{V} = \begin{bmatrix}\vec{v}_1&\ldots & \vec{v}_K \end{bmatrix}\in \set{C}^{M\times K}$ for enhancing signal quality and for mitigating interference. \begin{remark}We recall that $\set{C}$ denotes the set of complex-valued random variables, and hence that $\vec{H}$ and $\vec{V}$ are complex-valued random matrices. The explicit distinction between random and deterministic quantities is particularly important for the presentation of the main results in this study. \end{remark}

\paragraph*{Downlink}
The received signal of user $k\in \set{K}$ is given by
\begin{equation}\label{eq:dl}
y_k = \sum_{j\in \set{K}} \vec{h}_k^\herm \vec{v}_jx_{j} + z_k,
\end{equation}
where $z_k\sim \CN(0,1)$ is additive white Gaussian noise, and $x_j\sim \CN(0,1)$ is the independent data-bearing signal transmitted to user $j\in \set{K}$ using the linear precoding vector (beamformer) $\vec{v}_j \in \set{C}^M$. 

\paragraph*{Uplink}
The received signal at all antennas is given by
\begin{equation}\label{eq:ul}
\vec{r} = \sum_{k\in \set{K}} \sqrt{p_k}\vec{h}_kx_k + \vec{z},
\end{equation}
where $\vec{z}\sim \CN(\vec{0},\vec{I}_M)$ is additive white Gaussian noise, and $x_k\sim \CN(0,1)$ is the independent data-bearing signal transmitted by user $k\in \set{K}$ with power scaling coefficient $p_k \in \set{P}$, where we recall that $\set{P}$ denotes the set of nonnegative random variables. The estimate of the signal of user $k\in\set{K}$ is then given by $y_k = \vec{v}_k^\herm\vec{r}$, where $\vec{v}_k \in \set{C}^M$ is the corresponding linear combining vector (beamformer). For convenience, we also define the power vector $\vec{p}\eqdef [p_1,\ldots,p_K]^\T$.

For both models, we then focus on a well-defined notion of achievable rates for channels with state and receiver CSI \cite{caire1999capacity} \cite[Chapter~7]{el2011network}, often referred to as \textit{ergodic} achievable rates in the context of fading channels \cite{caire2018ergodic}. Specifically, we focus on the tuple of simultaneously achievable rates (in bit/symbol) given by the conditional mutual information  
\begin{equation}\label{eq:mutual_info}
(\forall k \in \set{K})~R_k \eqdef I\left(x_k;y_k \middle|U_k \right),
\end{equation}
where $U_k$ models some side information about $(\vec{H},\vec{V})$ available at the decoder of the signal of user $k \in \set{K}$.

\subsection{Ergodic rate bounds}
Although the mutual information in \eqref{eq:mutual_info} could be rather accurately approximated by means of numerical techniques, many upper and lower bounds have been developed in the literature to simplify simulations and, perhaps most importantly, the optimization of system parameters such as the beamformers and the power control policy. The most popular ones are reviewed below, using the unified compact notation\footnote{With an abuse of notation, in the following we give SINR expressions using the symbol $\frac{a}{b}$ also for the undefined case $(a,b)=(0,0)$, arising for zero-valued uplink beamformers, by tacitly assuming the convention $\frac{0}{0}=0$.}
\begin{equation*}
(\forall (j,k) \in \set{K}^2)~g_{j,k} \eqdef \begin{cases}
\sqrt{p_j}\vec{h}_j^\herm \vec{v}_k &\text{if uplink},\\
\vec{h}_k^\herm \vec{v}_j &\text{if downlink},
\end{cases}
\end{equation*}
\begin{equation*}
(\forall k \in \set{K})~\sigma_k^2 \eqdef \begin{cases}
\|\vec{v}_k\|^2 &\text{if uplink},\\
1 & \text{if downlink}.
\end{cases}
\end{equation*}
Note that we do not normalize the beamformers to $\|\vec{v}_k\|=1$, as this is not an information lossless procedure in general.

The first bound is an upper bound, and it will not be used in our main theoretical results as it is intractable to beamforming optimization under imperfect CSI. Nevertheless, it is a useful reference for evaluating the tightness of lower bounds, and a good estimate of the performance of practical systems \cite{gottsch2023subspace}.
\begin{proposition}[Optimistic ergodic rates]
The following upper bound holds $(\forall k \in \set{K})$
\begin{equation}\label{eq:oer}
\begin{gathered}
R_k \leq R_k^{\mathsf{oer}} \eqdef \E \left[\log\left(1+\mathsf{SINR}_k^{\mathsf{oer}}\right) \right], \\
\mathsf{SINR}_k^{\mathsf{oer}} \eqdef \dfrac{|g_{k,k}|^2 }{\sum_{j\in \set{K}\backslash \{k\}}|g_{j,k}|^2+\sigma_k^2}.
\end{gathered}
\end{equation}
\end{proposition}
\begin{proof}
Standard. See, e.g., the proof of \cite[Lemma~1]{caire2018ergodic}.
\end{proof}

The second bound is one of the tightest known lower bounds and, as we will see, is also amenable to beamforming optimization under certain conditions related to centralized or semi-distributed beamforming architectures.
\begin{proposition}[Coherent decoding lower bound]\label{prop:coh}
The following lower bound holds $(\forall k \in \set{K})$
\begin{equation}\label{eq:coh}
\begin{gathered}
R_k \geq R_k^{\mathsf{coh}} \eqdef \E \left[\log\left(1+\mathsf{SINR}_k^{\mathsf{coh}}\right) \right], \\
\mathsf{SINR}_k^{\mathsf{coh}} \eqdef \dfrac{|\E[g_{k,k}|U_k]|^2 }{\V(g_{k,k}|U_k)+\underset{j\in \set{K}\backslash \{k\}}{\sum}\E[|g_{j,k}|^2|U_k]+\E[\sigma_k^2|U_k]}.
\end{gathered}
\end{equation}
\end{proposition}
\begin{proof}
Simple variation of the UatF / hardening bounding technique, see, e.g., \cite[Appendix~C.2.6]{marzetta2016fundamentals} or \cite[Eq.~(36)]{caire2018ergodic}. We give a slightly different proof in Appendix~\ref{app:coh}, since instrumental for the novel results of this paper.
\end{proof}

The third bound is extensively used in the massive MIMO literature to obtain closed-form expressions, to perform statistical power control, and as an (asymptotically in $M$) accurate estimate of the performance in the absence of decoder CSI \cite{marzetta2016fundamentals,massivemimobook,caire2018ergodic}. Furthermore, as we will see, it is amenable to beamforming optimization in very general cases.
\begin{proposition}[UatF or hardening lower bound]\label{prop:uatf}
The following lower bound holds $(\forall k \in \set{K})$
\begin{equation}\label{eq:uatf}
\begin{gathered}
R_k^{\mathsf{coh}} \geq R_k^{\mathsf{UatF/hard}} = \log(1+\mathsf{SINR}_k^\mathsf{UatF/hard}), \\
\mathsf{SINR}_k^\mathsf{UatF/hard} = \dfrac{|\E[g_{k,k}]|^2 }{\V(g_{k,k})+\sum_{j\in \set{K}\backslash \{k\}}\E[|g_{j,k}|^2]+\E[\sigma_k^2]}.
\end{gathered}
\end{equation}
\end{proposition}
\begin{proof}
The proof is given in Appendix~\ref{app:uatf}.
\end{proof}

\begin{remark}
The bound $(\forall k \in \set{K})~R_k\geq R_k^{\mathsf{UatF/hard}}$ is well-known \cite{marzetta2016fundamentals,massivemimobook,caire2018ergodic}. In contrast, despite being intuitive, the relation between the two bounds $(\forall k \in \set{K})~R_k^{\mathsf{coh}}\geq R_k^{\mathsf{UatF/hard}}$ given by Proposition~\ref{prop:uatf} was not reported in the literature, to our best knowledge.
\end{remark}
Note that $R_k^{\mathsf{UatF/hard}}$ coincides with $R_k^{\mathsf{coh}}$ in the absence of decoder CSI, i.e., if the side information $U_k$ in \eqref{eq:mutual_info} is independent of $(\vec{H},\vec{V})$.

\subsection{Information constraints}\label{ssec:info}
As customary, we assume that $\vec{v}_k = \vec{v}_k(S_k)$ is a function of some CSI $S_k$ about $\vec{H}$ (e.g., noisy measurements of $\vec{H}$) available at the infrastructure side. More formally, we let $(\forall k \in \set{K})$ $\vec{v}_k \in \set{F}^M_{S_k}$, where $\set{F}^M_{S_k}\subseteq \set{C}^M$ denotes the set of $M$-tuples of functions of $S_k$ (i.e., each entry of $\vec{v}_k$ is a complex-valued function of $S_k$). Note that this CSI may not necessarily coincide with the decoder CSI $U_k$, especially in the downlink or in distributed uplink processing architectures. 

Furthermore, to model practical aspects related to scalable beamforming architectures in large and potentially cell-free systems, we assume that the entries of $\vec{v}_k$ may depend on different CSI, and that some entries of $\vec{v}_k$ may be forced to zero, i.e., that not all antennas are used to serve all users. We follow the general approach introduced in \cite{miretti2021team,miretti2024duality}, which allows to model the above aspects in a unified way by means of so-called  \textit{information constraints}\footnote{This terminology refers to the measure-theoretic notion of information used, e.g., in stochastic control theory \cite{yukselbook}.} $\vec{v}_k \in \set{V}_k$, where $\set{V}_k\subseteq \set{F}_{S_k}^M$ denotes a certain subset of the set of $M$-tuples of functions of $S_k$, defined as follows. We first partition the $M=NL$ antennas into $L$ groups of $N$ antennas each, whose local channels are given by the corresponding $N$ rows of 
\begin{equation*}
\vec{H} \reqdef \begin{bmatrix}
\vec{H}_1 \\ \vdots \\ \vec{H}_{L}
\end{bmatrix} \reqdef \begin{bmatrix}
\vec{h}_{1,1} & \ldots & \vec{h}_{1,K} \\ \vdots & \ddots &\vdots \\ \vec{h}_{L,1} & \ldots & \vec{h}_{L,K}
\end{bmatrix} \in \set{C}^{NL\times K},
\end{equation*}
and assume that the corresponding entries of the joint beamformer $\vec{v}_k = \begin{bmatrix}
\vec{v}_{1,k}^\herm & \ldots & \vec{v}_{L,k}^\herm
\end{bmatrix}^\herm$ of each user $k\in \set{K}$ are computed by $L$ separate processing units. In the context of cell-free massive MIMO, these processing units can be directly mapped to access points. We further let each user $k\in \set{K}$ be jointly served by a subset $\set{L}_k \subseteq \set{L} = \{1,\ldots,L\}$ of these units. 

\begin{definition}[Information constraints]Given an arbitrary CSI tuple $S_k=(S_{1,k},\ldots,S_{L,k})$, where $S_{l,k}$ is the CSI used by the $l$th processing unit to compute the beamformer $\vec{v}_{l,k}$ of the $k$th user, we define the information constraints $(\forall k \in \set{K})$
\begin{equation}\label{eq:distributedCSI}
		\set{V}_k \eqdef \set{V}_{1,k}\times \ldots \times \set{V}_{L,k}, \quad \set{V}_{l,k} \eqdef \begin{cases}
\set{F}^{N}_{S_{l,k}} & \text{if } l \in \set{L}_k,\\
\{\vec{0}\} & \text{otherwise},
\end{cases}
\end{equation}
where $\set{F}^{N}_{S_{l,k}}\subseteq \set{F}^N_{S_k}$ denotes the set of $N$-tuples of functions of $S_{l,k}$.
\end{definition}
The interested reader is referred to \cite[Chapter~2]{yukselbook} for additional details on the above definitions\footnote{In this work, we do not consider a second-order moment constraint  as in \cite{miretti2021team,miretti2024duality}, and we adhere closely to the more general model in \cite[Chapter~2]{yukselbook}.}. Nevertheless, we stress that the above constraints cover all common models in the user-centric cell-free massive MIMO literature, including the case of (clustered) centralized beamforming architectures based on global CSI, the case of (clustered) distributed beamforming architectures based on local CSI \cite{demir2021}, as well as more involved intermediate cases \cite{miretti2021team,miretti2024delayed}. For example, $S_{l,k}$ could be composed by \textit{local} noisy measurements of $\vec{H}_l$, and by measurements of the other local channels $\vec{H}_j$ exchanged over the fronthaul. We will give some practical examples throughout the manuscript.

We conclude this section by defining the following auxiliary notation, which will be used in many parts of this study to formalize assumptions related to the availability of some common information for beamforming.
\begin{definition}[Common beamforming CSI]\label{def:common_CSI} For all $k\in \set{K}$ and $l\in \set{L}_k$,  we assume without loss of generality that  $S_{l,k}$ takes the form $(S_{l,k}',Z_k)\eqdef S_{l,k}$, where $Z_k$ denotes some common beamforming CSI available at all processing units in $\set{L}_k$, and where $S_{l,k}'$ denotes the remaining parts of the CSI $S_{l,k}$.
\end{definition}
Clearly, the case of no common beamforming CSI is also covered by the above assumption, e.g., by choosing $Z_k$ to be deterministic or independent of everything else. Note that, at this point, we do not impose any additional assumption on $Z_k$, beyond it being shared by the processing units. In particular, $Z_k$ can be chosen arbitrarily and does not need to model the entire common information, e.g., the $S_{l,k}'$ may still be correlated. 
 
\section{Optimality of uplink MSE processing}\label{sec:uplink}
In this section, we review and provide new insights on the optimality of the MSE criterion for beamforming optimization. While ubiquitous in the classical multi-user MIMO literature, this criterion was rigorously connected to ergodic achievable rates only in terms of $R_k^{\mathsf{oer}}$ in \eqref{eq:oer} under perfect CSI (i.e., $\set{V}_k = \set{C}^M$), or of a special form of $R_k^{\mathsf{coh}}$ in \eqref{eq:coh} restricted to centralized beamforming, Gaussian fading, and MMSE channel estimation (see, e.g., \cite[Section~4.1]{massivemimobook}). In contrast, we will review a recently identified general connection with $R_k^{\mathsf{UatF/hard}}$ in \eqref{eq:uatf}, and derive a novel connection to $R_k^{\mathsf{coh}}$ in \eqref{eq:coh} which unifies and extends all previous results.    

\subsection{Use-and-then-forget lower bound}\label{sec:uplink_UatF}
For every $k\in \set{K}$ and $\vec{v}_k\in \set{V}_k$, let $\mathsf{SINR}_k^{\mathsf{UatF}}(\vec{v}_k)$ be the SINR in \eqref{eq:uatf} specialized to the uplink case, where the dependency on the beamformers is made explicit. We stress that, in contrast to the downlink case (covered later in Section~\ref{sec:duality}), the uplink SINR of each user $k\in \set{K}$ does not depend on the beamformers $\vec{v}_j$ used to combine the signals of the other users $j\in \set{K}\backslash \{k\}$. Furthermore, let us consider the MSE between the uplink signals $x_k$ and $y_k$ of user~$k\in\set{K}$, i.e.,
\begin{equation}\label{eq:MSE}
\begin{aligned}
(\forall \vec{v}_k\in \set{V}_k)~\mathsf{MSE}_k(\vec{v}_k)&= \E[|x_k - y_k|^2].
\end{aligned}
\end{equation} 

The following properties rigorously connect $\mathsf{SINR}_k^{\mathsf{UatF}}(\vec{v}_k)$ and $\mathsf{MSE}_k(\vec{v}_k)$, and (under slightly different model assumptions) they have been first identified in \cite{miretti2021team,miretti2024duality}.
\begin{proposition}
\label{prop:MMSE-SINR}
For all $k\in\set{K}$,
\begin{equation*}
1+ \sup_{\vec{v}_k\in \set{V}_k} \mathsf{SINR}_k^{\mathsf{UatF}}(\vec{v}_k) = \Big(\inf_{\vec{v}_k\in \set{V}_k}\mathsf{MSE}_k(\vec{v}_k)\Big)^{-1}.
\end{equation*}
Furthermore, if $\vec{v}_k^\star \in \arg\inf_{\vec{v}_k\in \set{V}_k}\mathsf{MSE}_k(\vec{v}_k)$, then  $\vec{v}_k^\star\in \arg \sup_{\vec{v}_k\in \set{V}_k} \mathsf{SINR}_k^{\mathsf{UatF}}(\vec{v}_k)$. 
\end{proposition}
\begin{proof} Particular case of Proposition~\ref{prop:MMSE-SINR-coh} with constant $U_k$.  
\end{proof}
Note that the converse of the second statement does not hold in general, since any phase rotated version $\vec{v}_k^\star e^{j\theta}$ ($\theta\in \stdset{R}$) of $\vec{v}_k^\star\in \arg\sup_{\vec{v}_k\in \set{V}_k} \mathsf{SINR}_k^{\mathsf{UatF}}(\vec{v}_k)$ gives the same optimal SINR, but may not (and typically does not) minimize the MSE. 

Moreover, by the monotonicity of the logarithm, we readily obtain the following corollary.
\begin{corollary}\label{cor:uatf}
For all $k\in\set{K}$, if $\vec{v}_k^\star\in \arg\inf_{\vec{v}_k\in \set{V}_k}\mathsf{MSE}_k(\vec{v}_k)$, then $\vec{v}_k^\star\in \arg\sup_{\vec{v}_k \in \set{V}_k} R_k^{\mathsf{Uatf},\mathsf{ul}}(\vec{v}_k)$. 
\end{corollary} 

To study the properties of $\inf_{\vec{v}_k\in \set{V}_k}\mathsf{MSE}_k(\vec{v}_k)$, it is convenient to rewrite the MSE as the following expectation of strongly convex quadratic forms $(\forall k \in \set{K})$
\begin{equation*}
\begin{aligned}
\mathsf{MSE}_k(\vec{v}_k)= \E[\vec{v}_k^\herm(\vec{H}\vec{P}\vec{H}^\herm + \vec{I}_M)\vec{v}_k -2\sqrt{p_k}\Re(\vec{h}_k^\herm\vec{v}_k)+1],
\end{aligned}
\end{equation*}
where the equality follows from simple algebraic manipulations and by defining $\vec{P}=\mathrm{diag}(\vec{p})$. This allows to exploit known theoretical results on the minimization of expected quadratic forms over $\vec{v}_k\in\set{V}_k$. For example, optimal centralized beamformers could be readily obtained by minimizing the conditional MSE pointwise for each realization of $S_k$. More generally, borrowing elements from the control theoretical literature on multi-agent coordinated decision making under partial information (also called  \textit{team} decision theory \cite{yukselbook}), the structure of optimal beamformers can be conveniently characterized by a set of necessary and sufficient optimality conditions \cite{miretti2021team}, as formalized in the following. 

\begin{proposition}\label{prop:stationarity}
For all $k\in \set{K}$, assume $(\forall l\in \set{L}_k)~\vec{p}\in \set{P}^K_{S_{l,k}}$ and the regularity condition $\E[\|\vec{H}\vec{H}^\herm\|_\mathrm{F}^2]< \infty$. Then $\vec{v}_k^\star\in \arg\inf_{\vec{v}_k\in \set{V}_k}\mathsf{MSE}_k(\vec{v}_k)$ if and only if $\vec{v}_k^\star$ is the unique solution to the feasibility problem
\begin{align}\label{eq:stationarity}
&\textnormal{find} ~ \vec{v}_k\in \set{V}_k ~ \textnormal{such that} ~(\forall l\in \set{L}_k) \notag\\
\begin{split}
&\vec{v}_{l,k} =\Big(\E[\vec{H}_l\vec{P}\vec{H}_l^\herm|S_{l,k}]+\vec{I}_N\Big)^{-1}\\
&\Big(\E[\vec{H}_l|S_{l,k}]\vec{P}^{\frac{1}{2}}\vec{e}_k-\textstyle\sum_{j\in \set{L}_k\backslash \{l\}}\E[\vec{H}_l\vec{P}\vec{H}_j^\herm\vec{v}_{j,k}|S_{l,k}] \Big).
\end{split}
\end{align}
\end{proposition}
\begin{proof} The proof is given in Appendix~\ref{proof:stationarity}.
\end{proof}

\begin{remark}\label{rem:regularity}
The regularity condition $\E[\|\vec{H}\vec{H}^\herm\|_\mathrm{F}^2]< \infty$ is satisfied for any physically consistent (uniformly bounded) fading distribution, as well as for any Gaussian fading distribution.
\end{remark}

\begin{remark}\label{rem:power}
The assumption $(\forall l \in \set{L}_k)$ $\vec{p}\in \set{P}^K_{S_{l,k}}$ is equivalent to assuming that $\vec{p}$ is a function of some common beamforming CSI. This is readily satisfied for most uplink massive MIMO models, where powers are typically adjusted on a statistical basis, i.e., where $\vec{p}\in \stdset{R}_+^{K}$ is deterministic. Nevertheless, we consider this more general assumption because it is instrumental for the downlink case, covered later in Section~\ref{sec:duality}. 
\end{remark}

\begin{remark}\label{rem:centr}
For the fully centralized case $\set{V}_k = \set{F}_{S_k}^M$, optimal beamformers $\vec{v}_k^\star\in \arg\inf_{\vec{v}_k\in \set{V}_k}\mathsf{MSE}_k(\vec{v}_k)$ are given by 
\begin{equation*}
\vec{v}_k^\star = \Big(\E[\vec{H}\vec{P}\vec{H}^\herm|S_k]+\vec{I}_M\Big)^{-1}\E[\vec{H}\vec{P}^{\frac{1}{2}}|S_k]\vec{e}_k.
\end{equation*}
This result can be shown via direct arguments, or recovered by letting $L=1$ and $N=M$ in Proposition~\ref{prop:stationarity}. Similar expressions are readily given for the clustered centralized case, by considering only the submatrices $\vec{H}^{(k)}$ of $\vec{H}$ corresponding to the groups of antennas indexed by $\set{L}_k$.
\end{remark}

The main consequence of the above results is that optimizing the UatF lower bound under arbitrary beamforming architectures, and mild assumptions which are always satisfied in practice, is equivalent to finding a solution to the (infinite dimensional) feasibility problem \eqref{eq:stationarity}. For the particular case of (clustered) centralized beamforming architectures, optimal beamformers are readily obtained as described in Remark~\ref{rem:centr}. 
 
\subsection{Coherent decoding lower bound}
We now provide a novel generalization of the above properties to the coherent decoding lower bound. Specifically, for every $k\in \set{K}$ and $\vec{v}_k\in \set{V}_k$, let $\mathsf{SINR}_k^{\mathsf{coh,ul}}(\vec{v}_k|U_k)$ be the SINR in \eqref{eq:coh} specialized to the uplink case, where the dependency on the beamformers and the conditioning on the decoder CSI $U_k$ is made explicit. Note that, due to the conditioning on $U_k$, $\mathsf{SINR}_k^{\mathsf{coh,ul}}(\vec{v}_k|U_k)$ is a random variable. Similarly, let us consider the conditional MSE of user~$k\in\set{K}$  
\begin{equation}\label{eq:MSEcond}
\begin{aligned}
		(\forall \vec{v}_k\in \set{V}_k)~\mathsf{MSE}_k(\vec{v}_k|U_k)= \E[|x_k - y_k|^2|U_k].
\end{aligned}
\end{equation} 
In the following, to avoid mathematical technicalities arising with the optimization of random variables, the notion of supremum/infimum is replaced by the notion of \textit{essential} supremum/infimum, i.e., the random variable constructed by taking the supremum/infimum under the almost-sure ordering (see the notation section for a more precise definition). 

\begin{proposition}
\label{prop:MMSE-SINR-coh}
Assume $(\forall k\in\set{K})(\forall l\in \set{L}_k)$ $S_{l,k}=(S_{l,k}',U_k)$, i.e., that the decoder CSI $U_k$ in \eqref{eq:mutual_info} is also common beamforming CSI as in Definition~\ref{def:common_CSI}. For all $k\in\set{K}$, 
\begin{equation*}
1+ \esssup{\vec{v}_k\in \set{V}_k} \mathsf{SINR}_k^{\mathsf{coh,ul}}(\vec{v}_k|U_k) = \Big(\essinf{\vec{v}_k\in \set{V}_k}\mathsf{MSE}_k(\vec{v}_k|U_k)\Big)^{-1}.
\end{equation*}
Furthermore, if $\vec{v}_k^\star \in \arg\essinf{\vec{v}_k\in \set{V}_k}\mathsf{MSE}_k(\vec{v}_k|U_k)$, then $\vec{v}_k^\star \in \arg\esssup{\vec{v}_k\in \set{V}_k} \mathsf{SINR}_k^{\mathsf{coh,ul}}(\vec{v}_k|U_k)$.
\end{proposition}
\begin{proof} We first observe that $\essinf{\vec{v}_k\in\set{V}_k} \mathsf{MSE}_k(\vec{v}_k|U_k)$  
\begin{align*}
			&\overset{(a)}{=} \essinf{\vec{v}_k\in\set{V}_k} \essinf{\alpha \in \set{F}_{U_k}} \mathsf{MSE}_k(\alpha\vec{v}_k|U_k) \\
			&\overset{(b)}{=}  \essinf{\vec{v}_k\in\set{V}_k} \dfrac{1}{1+\mathsf{SINR}_k^{\mathsf{coh,ul}}(\vec{v}_k|U_k)} \\
&= \dfrac{1}{1+\esssup{\vec{v}_k\in\set{V}_k}\mathsf{SINR}_k^{\mathsf{coh,ul}}(\vec{v}_k|U_k)} 
\end{align*}
where $(a)$ follows from $(\forall \alpha\in \set{F}_{U_k})(\forall \vec{v}_k \in \mathcal{V}_k)$ $\alpha \vec{v}_k \in\mathcal{V}_k$, and $(b)$ follows from standard minimization of convex quadratic forms in $\stdset{C}$ for each realization of $U_k$ as in the proof of Proposition~\ref{prop:coh}. The second part of the statement follows from the similar chain of equalities $\mathsf{MSE}_k(\vec{v}_k^\star|U_k)= \essinf{\alpha \in \set{F}_{U_k}} \mathsf{MSE}_k(\alpha\vec{v}_k^\star|U_k) = (1+\mathsf{SINR}_k^{\mathsf{coh,ul}}(\vec{v}_k^\star|U_k))^{-1}$, where the first equality holds by assumption.
\end{proof}

\begin{corollary}\label{cor:coh}
Assume $(\forall k\in\set{K})(\forall l\in \set{L}_k)$ $S_{l,k}=(S_{l,k}',U_k)$, i.e., that $U_k$ is also common beamforming CSI. For all $k\in \set{K}$, if $\vec{v}_k^\star \in \arg\essinf{\vec{v}_k\in \set{V}_k}\mathsf{MSE}_k(\vec{v}_k|U_k)$, then $\vec{v}_k^\star \in \arg\sup_{\vec{v}_k \in \set{V}_k} R_k^{\mathsf{coh},\mathsf{ul}}(\vec{v}_k)$. 
\end{corollary}
The above results generalize the known related result in, e.g.,  \cite{massivemimobook,demir2021}, which is restricted to centralized beamforming, Gaussian fading, and MMSE channel estimation. Essentially, it states that the conditional SINR in \eqref{eq:coh} is maximized (in a stochastic sense) by the same beamformer that minimizes the conditional MSE, provided that the CSI $U_k$ used for decoding is also perfectly shared among all beamforming processing units serving user $k$. This assumption is clearly satisfied for centralized uplink processing architectures, regardless the channel model and estimation technique. Furthermore, it also covers semi-distributed architectures operating based on a mix of local and common CSI. An important remark is that, due to this assumption, optimal beamformers attaining the essential infimum can be constructed in practice by evaluating the pointwise infimum for each realization of $U_k$. Based on this intuition, we derive the following proposition that relates conditional and unconditional MSE minimization.  

\begin{proposition}\label{prop:ess}
For all $k\in \set{K}$, assume $(\forall l\in \set{L}_k)$ $S_{l,k}=(S_{l,k}',U_k)$, i.e., that $U_k$ is also common beamforming CSI. Furthermore, assume $(\forall l\in \set{L}_k)~\vec{p}\in \set{P}^K_{S_{l,k}}$, and the regularity condition $\E[\|\vec{H}\vec{H}^\herm\|_\mathrm{F}^2]< \infty$. Then, $(\forall k\in \set{K})$ $\vec{v}_k^\star \in \arg\inf_{\vec{v}_k\in \set{V}_k}\mathsf{MSE}_k(\vec{v}_k) \iff \vec{v}_k^\star \in\arg\essinf{\vec{v}_k\in \set{V}_k}\mathsf{MSE}_k(\vec{v}_k|U_k)$.
\end{proposition}
\begin{proof}
By definition of essential infimum, $(\forall \vec{v}_k\in \set{V}_k)$
\begin{equation*}
    \mathsf{MSE}_k(\vec{v}_k) = \E\left[\mathsf{MSE}_k(\vec{v}_k|U_k) \right] \geq \E\Big[\essinf{\vec{v}_k\in \set{V}_k}\mathsf{MSE}_k(\vec{v}_k|U_k) \Big]
\end{equation*}
holds. Hence, $\vec{v}_k^\star \in \arg \essinf{\vec{v}_k\in \set{V}_k}\mathsf{MSE}_k(\vec{v}_k|U_k) \implies \vec{v}_k^\star \in \arg \inf_{\vec{v}_k\in \set{V}_k}\mathsf{MSE}_k(\vec{v}_k)$ readily holds (without additional assumptions). The proof of the converse statement is given in Appendix~\ref{proof:ess}, and is based on the additional assumptions.
\end{proof}

\begin{remark}\label{rem:power2}
If $(\forall l\in \set{L}_k)$ $S_{l,k}=(S_{l,k}',U_k)$ and $\vec{p}\in \set{P}_{U_k}^K$, i.e., if $U_k$ is  also common information for both beamforming and power control, then $(\forall l \in \set{L}_k)$ $\vec{p}\in \set{P}^K_{S_{l,k}}$ readily holds. This observation is instrumental for the downlink case, covered later in Section~\ref{sec:duality}.
\end{remark}

Combining Corollary~\ref{cor:coh} and Proposition~\ref{prop:ess}, we then obtain the following connection between unconditional MSE minimization and ergodic rate maximization problems.
\begin{corollary}\label{cor:uatf-coh}
Assume $(\forall k \in \set{K})(\forall l\in \set{L}_k)$ $S_{l,k}=(S_{l,k}',U_k)$, i.e., that $U_k$ is also common beamforming CSI. Furthermore, assume $(\forall l\in \set{L}_k)~\vec{p}\in \set{P}^K_{S_{l,k}}$, and the regularity condition $\E[\|\vec{H}\vec{H}^\herm\|_\mathrm{F}^2]< \infty$. If $\vec{v}_k^\star \in \arg\inf_{\vec{v}_k\in \set{V}_k}\mathsf{MSE}_k(\vec{v}_k)$, then $\vec{v}_k^\star \in \arg\sup_{\vec{v}_k \in \set{V}_k} R_k^{\mathsf{coh},\mathsf{ul}}(\vec{v}_k)$. 
\end{corollary}
Interestingly, we observe that optimizing the unconditional MSE under centralized or semi-distributed beamforming architectures (and mild assumptions, see Remark~\ref{rem:regularity} and  Remark~\ref{rem:power}) is equivalent to optimizing both the coherent decoding lower bound and the UatF lower bound (see  Corollary~\ref{cor:uatf}). For example, if $S_k = U_k$, optimal (clustered) centralized beamformers in terms of both the coherent decoding lower bound and the UatF bound can be readily obtained as discussed in Remark~\ref{rem:centr}. Furthermore, more general cases can be addressed by finding a solution to the feasibility problem \eqref{eq:stationarity}.


\section{Joint channel estimation and beamforming}\label{sec:main}
Having established the optimality of uplink MSE processing under fairly general conditions, we now focus on the beamforming optimization problem for each user $k\in \set{K}$
\begin{equation}\label{eq:MSE_prob}
\begin{aligned}
\underset{\vec{v}_k\in \set{V}_k}{\text{minimize}}~\mathsf{MSE}_k(\vec{v}_k),
\end{aligned}
\end{equation} 
where we recall that $\set{V}_k$ denotes a subset of functions of some arbitrary CSI $S_k$. In particular, $S_k$ does not need to be restricted to channel estimates $\hat{\vec{H}}$ obtained as the output of a specific channel estimation procedure. For example, $S_k$ could be directly mapped to raw measurements of uplink pilot signals collected by the infrastructure. In this case, \eqref{eq:MSE_prob} should be rather interpreted as \textit{joint} channel estimation and beamforming optimization problems. In general, the usual approach of separating channel estimation and beamforming may not be optimal for the joint problem. Nevertheless, in the following we demonstrate that this approach is actually optimal under common assumptions that apply to a wide variety of models in the literature.  

\subsection{Centralized processing architectures}
We start by considering the particular case of centralized uplink processing architectures, for which an optimal solution can be obtained directly.
\begin{proposition}\label{prop:centr}
For all $k \in \set{K}$, let $(\vec{H},S_k)$ be circularly symmetric jointly Gaussian distributed, $(\forall l \in \set{L})$ $(\forall j \in \set{L})$ $S_{l,k} = S_{j,k}$ (centralized CSI), and $\vec{p}\in \set{P}^K_{S_k}$.  Then, the optimal solution to \eqref{eq:MSE_prob} is given by the centralized MMSE beamformer
\begin{equation}\label{eq:CMMSE}
\vec{v}_k^\star = \Big(\hat{\vec{H}}^{(k)}\vec{P}\hat{\vec{H}}^{(k)\herm} + \vec{\Psi}^{(k)}+\vec{I}_M  \Big)^{-1}\hat{\vec{H}}^{(k)}\vec{P}^{\frac{1}{2}}\vec{e}_k,
\end{equation}
where $\hat{\vec{H}}^{(k)} \eqdef \E[\vec{H}^{(k)}|S_k]$
is the (linear/affine) MMSE estimate of $\vec{H}^{(k)}\eqdef \vec{C}_k\vec{H}$ given $S_k$, where $\vec{C}_k \eqdef \mathrm{diag}(\vec{C}_{1,k},\ldots,\vec{C}_{L,k})$ is a block-diagonal matrix satisfying
\begin{equation*}
(\forall l \in \set{L})~\vec{C}_{l,k}\eqdef\begin{cases}
\vec{I}_N & \text{if } l\in \set{L}_k, \\
\vec{0}_{N\times N} & \text{otherwise},
\end{cases}
\end{equation*}
and where $\vec{\Psi}^{(k)} \eqdef \sum_{j\in \set{K}}p_j\E[\tilde{\vec{h}}_j^{(k)}\tilde{\vec{h}}_j^{(k)\herm}]$ collects the covariance matrices of the columns $\tilde{\vec{h}}_j^{(k)}$ of the estimation error $\tilde{\vec{H}}^{(k)}\eqdef \hat{\vec{H}}^{(k)}-\vec{H}^{(k)}$.
\end{proposition}
\begin{proof} Note that the assumptions of Proposition~\ref{prop:stationarity} are satisfied. For the case $\set{L}_k = \set{L}$ (and hence $\vec{H}^{(k)}=\vec{H}$), by recalling Remark~\ref{rem:centr}, we then have 
\begin{equation*}
\vec{v}_k^\star := \Big(\E[\vec{H}\vec{P}\vec{H}^\herm|S_k] + \vec{I}_M  \Big)^{-1}\E[\vec{H}|S_k]\vec{P}^{\frac{1}{2}}\vec{e}_k.
\end{equation*}
The conditional mean $\hat{\vec{H}}^{(k)}=\E[\vec{H}|S_k]$ is the MMSE estimate of $\vec{H}$ given $S_k$. Since $\vec{H}$ and $S_k$ are jointly Gaussian, then $\hat{\vec{H}}^{(k)}$ is also the linear/affine MMSE estimate of $\vec{H}$ given $S_k$, and the estimation error $\tilde{\vec{H}}^{(k)}=\vec{H}-\hat{\vec{H}}^{(k)}$ is independent of $S_k$. Therefore, by using $\E[\hat{\vec{H}}^{(k)}\vec{P}\tilde{\vec{H}}^{(k)\herm}|S_k]= \hat{\vec{H}}^{(k)}\vec{P}\E[\tilde{\vec{H}}^{(k)\herm}] = \vec{0}$, we have
\begin{align*}
\E[\vec{H}\vec{P}\vec{H}^\herm|S_k]&=\hat{\vec{H}}^{(k)}\vec{P}\hat{\vec{H}}^{(k)\herm} + \E[\tilde{\vec{H}}^{(k)}\vec{P}\tilde{\vec{H}}^{(k)\herm}|S_k]\\
&= \hat{\vec{H}}^{(k)}\vec{P}\hat{\vec{H}}^{(k)\herm} + \sum_{j\in \set{K}}p_j\E[\tilde{\vec{h}}_j^{(k)}\tilde{\vec{h}}_j^{(k)\herm}|S_k] \\
&= \hat{\vec{H}}^{(k)}\vec{P}\hat{\vec{H}}^{(k)\herm} + \sum_{j\in \set{K}}p_j\E[\tilde{\vec{h}}_j^{(k)}\tilde{\vec{h}}_j^{(k)\herm}].
\end{align*}
The more general case with $\set{L}_k \subseteq \set{L}$ follow readily by replacing $\vec{H}$ with $\vec{H}^{(k)}$ in all the above steps.
\end{proof}

The above proposition establishes that, under Gaussian measurements of Gaussian channels (and a mild assumption on the transmit powers, see Remark~\ref{rem:power}), optimal centralized uplink MSE processing can be decomposed into a linear/affine MMSE channel estimation step followed by a centralized MMSE beamforming step.

\begin{remark} 
Combined with Corollary~\ref{cor:uatf} and Corollary~\ref{cor:coh} in Section~\ref{sec:uplink}, Proposition~\ref{prop:centr} also establishes the optimality of this functional separation in terms of both the UatF lower bound and the coherent decoding lower bound (given decoder CSI $U_k=S_k$) on the achievable ergodic rates \eqref{eq:mutual_info}.
\end{remark}

Due to to its simplicity, this type of processing is already widely considered in the literature under specific models that satisfy our assumption on $(\vec{H},S_k)$ (see, e.g., \cite{massivemimobook},\cite{demir2021}). However, the optimality of this functional separation in terms of ergodic achievable rates was not reported before.

\subsection{Distributed processing architectures}
We now consider more general distributed uplink processing architectures subject to mild assumptions related to the CSI acquisition procedure. Specifically, we consider the following specialized version of the necessary and sufficient optimality conditions given by \eqref{eq:stationarity}.
\begin{proposition}\label{prop:distr}
For all $k \in \set{K}$, let $(\vec{H},S_k)$ be circularly symmetric jointly Gaussian distributed, $(\forall l \in \set{L})(\forall j \in \set{L}\backslash\{l\})$ $\vec{H}_l \to S_{l,k} \to S_{j,k} \to \vec{H}_j$ form a Markov chain, and $(\forall l\in \set{L}_k)~\vec{p}\in \set{P}^K_{S_{l,k}}$.
Then, Problem \eqref{eq:MSE_prob} admits a solution $\vec{v}^\star_k$ if and only if $\vec{v}_k^\star$ is the unique $\vec{v}_k\in \set{V}_k$ satisfying $(\forall l \in \set{L}_k)$
\begin{equation}\label{eq:TMMSE}
\vec{v}_{l,k} = \vec{V}_{l,k} \Bigg(\vec{e}_k - \sum_{j\in \set{L}_k \backslash \{l\}}\vec{P}^{\frac{1}{2}}\E\left[ \hat{\vec{H}}_j^{(k)\herm}\vec{v}_{j,k}\middle|S_{l,k}\right]  \Bigg),
\end{equation}
where $\vec{V}_{l,k} \eqdef \Big(\hat{\vec{H}}_l^{(k)}\vec{P}\hat{\vec{H}}_l^{(k)\herm }+ \vec{\Psi}_l^{(k)}+\vec{I}_N  \Big)^{-1}\hat{\vec{H}}_l^{(k)}\vec{P}^{\frac{1}{2}}$ is the local MMSE beamformer, where $\hat{\vec{H}}_l^{(k)} \eqdef \E[\vec{H}_l|S_{l,k}]$ is the (linear/affine) MMSE estimate of $\vec{H}_l$ given $S_{l,k}$, and where $\vec{\Psi}_l^{(k)} \eqdef \sum_{j\in \set{K}}p_j\E[\tilde{\vec{h}}_{l,j}^{(k)}\tilde{\vec{h}}_{l,j}^{(k)\herm}]$ collects the covariance matrices of the columns $\tilde{\vec{h}}_{l,j}^{(k)}$ of the local estimation error $\tilde{\vec{H}}_l^{(k)}\eqdef \hat{\vec{H}}_l^{(k)}-\vec{H}_l^{(k)}$. 
\end{proposition}
\begin{proof}
Note that the assumptions of Proposition~\ref{prop:stationarity} are satisfied. We then manipulate the expression
\begin{equation*}
\begin{split}
&\vec{v}_{l,k} =\Big(\E[\vec{H}_l\vec{P}\vec{H}_l^\herm|S_{l,k}]+\vec{I}_N\Big)^{-1}\\
&\Big(\E[\vec{H}_l|S_{l,k}]\vec{P}^{\frac{1}{2}}\vec{e}_k-\textstyle\sum_{j\in \set{L}_k\backslash \{l\}}\E[\vec{H}_l\vec{P}\vec{H}_j^\herm\vec{v}_{j,k}|S_{l,k}] \Big)
\end{split}
\end{equation*}
in \eqref{eq:stationarity} using the additional assumptions. For all $l\in \set{L}$, due to the Gaussianity of $(\vec{H}_l,S_{l,k})$, we have that $\hat{\vec{H}}^{(k)}_l=\E[\vec{H}_{l}|S_{l,k}]$ corresponds to the linear/affine MMSE estimate of $\vec{H}_l$ given $S_{l,k}$, and, by following similar steps as in the proof of Proposition~\ref{prop:centr}, that \begin{align*}
\E[\vec{H}_l\vec{P}\vec{H}_l^\herm|S_{l,k}]=\hat{\vec{H}}_l^{(k)}\vec{P}\hat{\vec{H}}_l^{(k)\herm} + \vec{\Psi}_l^{(k)}.
\end{align*}
Moreover, we have $(\forall l \in \set{L})(\forall j \in \set{L}\backslash\{l\})$
\begin{align*}
&\E[\vec{H}_l\vec{P}\vec{H}_j^\herm \vec{v}_{j,k} | S_{l,k}]\\
&= \E[\E[\vec{H}_l\vec{P}\vec{H}_j^\herm |S_{j,k},S_{l,k}]\vec{v}_{j,k} | S_{l,k}]\\
&= \E[\E[\vec{H}_l|S_{j,k},S_{l,k}]\vec{P}\E[\vec{H}_j^\herm |S_{j,k},S_{l,k}]\vec{v}_{j,k} | S_{l,k}]\\
&= \E[\E[\vec{H}_l|S_{l,k}]\vec{P}\E[\vec{H}_j^\herm |S_{j,k}]\vec{v}_{j,k} | S_{l,k}]\\
& = \hat{\vec{H}}^{(k)}_l\vec{P}\E[\hat{\vec{H}}_j^{(k)\herm}\vec{v}_{j,k} | S_{l,k}],
\end{align*}
where the first equality follows from the law of total expectation and $\vec{v}_{j,k} \in \set{F}_{S_{j,k}}^N$, the second equality from the Markov chain $\vec{H}_l \to (S_{l,k},S_{j,k}) \to \vec{H}_j$ and since $\vec{p}\in \set{P}^K_{S_{l,k}}$, and the third equality from the Markov chains $S_{l,k} \to S_{j,k} \to \vec{H}_j $ and $S_{j,k} \to S_{l,k} \to \vec{H}_l$. The proof is concluded by rearranging the terms.
\end{proof}
\begin{remark}
The Markov chain in the above proposition essentially corresponds to scenarios where each $j$th processing unit can only acquire a degraded version of the local measurements of the local channel $\vec{H}_l$ available at the $l$th processing unit. This is a very common setup in the distributed cell-free massive MIMO literature, where CSI is typically acquired by means of pilot-based local measurements of the local channels, potentially followed by an imperfect CSI sharing procedure over the fronthaul, and where the local channels are mutually independent due to their geographical separation.
\end{remark}

Proposition~\ref{prop:distr} establishes that, under degraded Gaussian measurements of Gaussian channels, optimal distributed uplink MSE processing can be decomposed into: (i) a local MMSE channel estimation stage; (ii) a local MMSE beamforming stage; and (iii) a correction stage for each processing unit that takes into account the impact of the other processing units based on the available information. 

\begin{remark}
Combined with Corollary~\ref{cor:uatf} in Section~\ref{sec:uplink}, Proposition~\ref{prop:distr} shows that this functional separation is also optimal in terms of the UatF lower bound on the achievable ergodic rates \eqref{eq:mutual_info}. Moreover, combined with Corollary~\ref{cor:coh}, it also establishes the optimality in terms of the coherent decoding lower bound, provided that the decoder CSI $U_k$ is perfectly shared among all processing units $l\in \set{L}_k$ serving user $k\in \set{K}$. 
\end{remark}

We point out that very similar optimality conditions to those in Proposition~\ref{prop:distr} were already reported in \cite{miretti2021team}[Lemma~2], by restricting the CSI of each processing unit to local (MMSE) channel estimates and to degraded versions of the local estimates available at the other processing units. Our result considers  more general setups.

\subsection{Examples}\label{sec:examples}
The optimality conditions in Proposition~\ref{prop:distr} correspond to an infinite dimensional linear system of equations. In general, solving this system may require approximate numerical algorithms \cite{yukselbook}. Nevertheless, in the following we present two simple yet practical examples that admit a closed-form solution in terms of channel statistics. First, we consider the case of fully local processing.
\begin{proposition}\label{prop:local} For all $k \in \set{K}$, let $(\forall l \in \set{L})(\forall j \in \set{L}\backslash\{l\})$ $(\vec{H}_l,S_{l,k})$ be circularly symmetric jointly Gaussian distributed and independent of $(\vec{H}_j,S_{j,k})$ (local CSI), and $\vec{p}\in \stdset{R}_+^K$. Then, the optimal solution to \eqref{eq:MSE_prob} is given by
\begin{equation}\label{eq:LTMMSE}
(\forall l\in \set{L})~\vec{v}_{l,k}^\star = \vec{V}_{l,k}\vec{c}_{l,k}, 
\end{equation}
where $\vec{V}_{l,k}$ is the local MMSE stage obtained as a function of the local MMSE channel estimate $\hat{\vec{H}}^{(k)}_l$ defined in Proposition~\ref{prop:distr}, and  $\vec{c}_{l,k}\in \stdset{C}^K$ is a statistical beamforming stage given by the unique solution to the linear system of equations
\begin{equation*}
\begin{cases}\vec{c}_{l,k} + \sum_{j \in \set{L}_k \backslash \{l\}}\vec{\Pi}_{j,k} \vec{c}_{j,k} = \vec{e}_k & \forall l \in \set{L}_k, \\
\vec{c}_{l,k} = \vec{0}_{K\times 1} & \text{otherwise,}
\end{cases}
\end{equation*}
where $(\forall l\in \set{L})~\vec{\Pi}_{l,k} \eqdef \E\left[\vec{P}^{\frac{1}{2}}\hat{\vec{H}}_l^{(k)\herm}\vec{V}_{l,k}\right]$.
\end{proposition}
\begin{proof}
The proof follows by verifying that \eqref{eq:LTMMSE} satisfies the optimality conditions \eqref{eq:TMMSE}. The details are similar to \cite[Theorem~4]{miretti2021team}, and are reported in Appendix~\ref{app:local}.
\end{proof}
The solution in \eqref{eq:LTMMSE} corresponds to the local \textit{team} MMSE solution derived in \cite{miretti2021team,miretti2024duality} in the context of cell-free massive MIMO syestems by restricting the CSI of each access point to local MMSE channel estimates ($S_{l,k} = \hat{\vec{H}}^{(k)}_l$). Our result demonstrates that, for arbitrary local measurements (e.g., raw local observations of pilot signals), the optimal local processing can be decomposed into local MMSE channel estimation followed by local team MMSE beamforming. We recall that, as observed in \cite{miretti2021team,miretti2024duality}, under specific assumptions that can be mapped to zero-mean channels (Rayleigh fading) and no pilot contamination, the local team MMSE solution boils down to the known local MMSE scheme with optimal large-scale fading decoding \cite{demir2021}. Nevertheless, the former scheme can be significantly outperformed by the latter scheme when these assumptions are violated, e.g., under Ricean fading \cite{noor2025wcnc}. We refer to \cite{miretti2021team,miretti2024duality} for additional practical details on \eqref{eq:LTMMSE}.

\begin{remark}
The local team MMSE solution in \eqref{eq:LTMMSE} maximizes the UatF bound, but not the coherent decoding lower bound in most cases. This is because the independence of $S_{l,k}$ and $S_{j,l}$ implies that $U_k$ cannot be common beamforming CSI as required by Corollary \ref{cor:coh}, except for trivial cases such as deterministic or uninformative $U_k$.
\end{remark}

As second example, we consider a simple extension of the local team MMSE solution to the case of semi-distributed processing based on mixed local and common CSI.
\begin{proposition}\label{prop:mixed} For all $k \in \set{K}$, let $(\forall l \in \set{L})(\forall j \in \set{L}\backslash\{l\})$ $(\vec{H}_l,S_{l,k})=(\vec{H}_l,S_{l,k}',Z_k)$ be circularly symmetric jointly Gaussian distributed and conditionally independent of $(\vec{H}_j,S_{j,k})=(\vec{H}_j,S_{j,k}',Z_k)$ given some common CSI $Z_k$ (mixed local and common CSI), and $\vec{p}\in \set{F}_{Z_k}^K$. Then, the optimal solution to \eqref{eq:MSE_prob} is given by
\begin{equation}\label{eq:MTMMSE}
(\forall l\in \set{L})~\vec{v}_{l,k}^\star = \vec{V}_{l,k}\vec{c}_{l,k}, 
\end{equation}
where $\vec{V}_{l,k}$ is the local MMSE stage obtained as a function of the local MMSE channel estimate $\hat{\vec{H}}^{(k)}_l$ defined in Proposition~\ref{prop:distr}, and  $\vec{c}_{l,k}\in \set{F}^K_{Z_k}$ is a centralized beamforming stage given by the unique solution to the linear system of equations
\begin{equation*}
\begin{cases}\vec{c}_{l,k} + \sum_{j \in \set{L}_k \backslash \{l\}}\vec{\Pi}_{j,k} \vec{c}_{j,k} = \vec{e}_k & \forall l \in \set{L}_k, \\
\vec{c}_{l,k} = \vec{0}_{K\times 1} & \text{otherwise,}
\end{cases}
\end{equation*}
where $(\forall l\in \set{L})~\vec{\Pi}_{l,k} \eqdef \E\left[\vec{P}^{\frac{1}{2}}\hat{\vec{H}}_l^{(k)\herm}\vec{V}_{l,k}\middle|Z_k\right]$.
\end{proposition}
\begin{proof}(Sketch)
The proof follows similar steps as the proof of Proposition~\ref{prop:local} by replacing the deterministic matrices $\E[\vec{P}^{\frac{1}{2}}\hat{\vec{H}}_l^{(k)\herm}\vec{V}_{l,k}]$ with the random matrices $\E[\vec{P}^{\frac{1}{2}}\hat{\vec{H}}_l^{(k)\herm}\vec{V}_{l,k}|Z_k]$ everywhere, and hence it is omitted.
\end{proof}
Note that, in contrast to \eqref{eq:LTMMSE}, $\vec{c}_{l,k}$ in \eqref{eq:MTMMSE} is a function of the common CSI $Z_k$, and hence it is a random vector. A potential application of the above result is to let $Z_k$ model the output of a measurement sharing procedure impaired by delay, as done in \cite{miretti2024delayed}. Other applications may include the modeling of $Z_k$ as shared low-dimensional projections (linear transformations) of the local measurements. The derivation and analysis of such novel information sharing patterns and corresponding optimal semi-distributed beamforming schemes is out of the scope of this work, and left as a promising future research direction.  

\begin{remark}
The mixed local/centralized team MMSE solution in \eqref{eq:MTMMSE} maximizes the UatF bound and, if the decoder CSI coincides with the common information for beamforming, i.e., if $U_k = Z_k$, also the coherent decoding lower bound. 
\end{remark}

We conclude this section by pointing out that other closed-form expressions could be obtained also for more exotic sequential information sharing procedures over serial fronthauls, for example by generalizing the results in \cite{miretti2021team} \cite{miretti2021team2} using Proposition~\ref{prop:distr}. We omit the details due to space limitations.

\section{Optimality of downlink MSE processing}\label{sec:duality}
In this section we show that the main findings of the previous sections, derived for the uplink case, can be extended to the downlink case by leveraging various forms of duality between uplink and downlink achievable ergodic rate regions. 

\subsection{Hardening lower bound}\label{sec:duality_hard}
For every $k\in \set{K}$ and $(\vec{v}_1,\ldots,\vec{v}_K)\in \set{V}_1\times \ldots \times \set{V}_K$, let $\mathsf{SINR}_k^{\mathsf{hard}}(\vec{v}_1,\ldots,\vec{v}_K)$ be the SINR in \eqref{eq:uatf} specialized to the downlink case, where the dependency on the beamformers is made explicit. Similarly, for every $k\in \set{K}$, $\vec{v}_k\in \set{V}_k$, and $\vec{p}\in \stdset{R}^K_{+}$, we let $\mathsf{SINR}_k^{\mathsf{UatF}}(\vec{v}_k,\vec{p})$ be the SINR in \eqref{eq:uatf} specialized to the uplink case. Given some ergodic rate constraints, or, equivalently, some SINR constraints $(\gamma_1,\ldots,\gamma_K)\in \stdset{R}^K_{+}$, we consider the following downlink feasibility problem 
\begin{equation}\label{prob:feas_hard}
\begin{aligned}
\text{find} \quad & (\vec{v}_1,\ldots,\vec{v}_K,\vec{p})\in \set{V}_1\times \ldots \times \set{V}_K\times \stdset{R}_{+}^K\\
\text{such that} \quad & (\forall k \in \set{K})~\mathsf{SINR}_k^{\mathsf{hard}}(\sqrt{p_1}\vec{v}_1,\ldots,\sqrt{p_K}\vec{v}_K) = \gamma_k,
\end{aligned}
\end{equation}
and the uplink feasibility problem
\begin{equation}\label{prob:feas_UatF}
\begin{aligned}
\text{find} \quad & (\vec{v}_1,\ldots,\vec{v}_K,\vec{p})\in \set{V}_1\times \ldots \times \set{V}_K\times \stdset{R}_{+}^K\\
\text{such that} \quad & (\forall k \in \set{K})~\mathsf{SINR}_k^{\mathsf{UatF}}(\vec{v}_k,\vec{p}) = \gamma_k.
\end{aligned}
\end{equation}
Note that, in the downlink formulation \eqref{prob:feas_hard}, the total power allocated to user $k\in\set{K}$ is given by $p_k\E[\|\vec{v}_k\|^2]$, i.e., the beamformers are not normalized.
The above problems can be rigorously connected by means of well-known uplink-downlink duality arguments.

\begin{proposition}\label{prop:duality_hard}
Problem~\eqref{prob:feas_hard} admits a solution $(\vec{v}_1^\star,\ldots,\vec{v}_K^\star,\vec{p}^\mathsf{dl})$ if and only if Problem~\eqref{prob:feas_UatF} admits a solution $(\vec{v}_1^{\star},\ldots,\vec{v}_K^{\star},\vec{p}^\mathsf{ul})$, for the same beamformers $(\vec{v}_1^\star,\ldots,\vec{v}_K^\star)$ ($\vec{p}^\mathsf{ul}$ and $\vec{p}^\mathsf{dl}$ can be different). Moreover, the power vector $\vec{p}^{\mathsf{dl}}$ (resp. $\vec{p}^{\mathsf{ul}}$) solves the system $(\vec{D}-\vec{\Gamma}\vec{B})\vec{p}^{\mathsf{dl}} = \vec{\Gamma}\vec{1}$ (resp. $(\vec{D}-\vec{\Gamma}\vec{B}^\T)\vec{p}^{\mathsf{ul}} = \vec{\Gamma}\vec{\Sigma}\vec{1}$) parametrized by $\vec{\Gamma}\eqdef \mathrm{diag}\left(\gamma_1,\ldots,\gamma_K\right)$, $\vec{D}\eqdef \mathrm{diag}\left(|\E[\vec{h}^\herm_1\vec{v}^{\star}_1]|^2,\ldots,|\E[\vec{h}^\herm_K\vec{v}^{\star}_K]|^2\right)$, $\vec{\Sigma}\eqdef \mathrm{diag}\left(\E[\|\vec{v}_1^{\star}\|^2],\ldots,\E[\|\vec{v}_K^\star\|^2]\right)$, and
\begin{align*}
\vec{B}&\eqdef \begin{bmatrix}
\E[|\vec{h}_1^\herm\vec{v}_1^{\star}|^2] & \ldots & \E[|\vec{h}_1^\herm\vec{v}^{\star}_K|^2] \\
\vdots & \ddots & \vdots \\
\E[|\vec{h}_K^\herm\vec{v}^{\star}_1|^2] & \ldots & \E[|\vec{h}_K^\herm\vec{v}^{\star}_K|^2]
\end{bmatrix}-\vec{D}.
\end{align*} 
Furthermore, the two solutions satisfy $\vec{1}^\T \vec{p}^\mathsf{ul}= \vec{1}^\T \vec{\Sigma}\vec{p}^\mathsf{dl}$.
\end{proposition}
\begin{proof}
The proof follows from known uplink-downlink duality arguments in the power control literature (reviewed, e.g., in \cite{schubert2004solution, massivemimobook}). For completeness, a proof tailored to the considered system model and notation is reported in Appendix~\ref{proof:duality_hard}. 
\end{proof}

An important consequence of Proposition~\ref{prop:duality_hard} is that downlink beamforming optimization problems involving the hardening bound and sum power constraints can be equivalently performed over a (virtual) dual uplink channel. Popular examples include weighted sum-rate maximization subject to a sum power constraint, and weighted minimum SINR maximization subject to a sum power constraint. 
It is important to clarify that the main benefit of working on a dual uplink channel is not to turn difficult problems into easier ones. In particular, although downlink weighted minimum SINR maximization can be conveniently addressed using known methods for the uplink case \cite{miretti2023fixed}, provably NP-hard problems such as downlink weighted sum-rate maximization remain equally difficult on a dual uplink channel. However, working on a dual uplink channel helps identifying useful structural properties of the solution, as discussed in more detail below.

All the aforementioned optimization problems typically seek for Pareto optimal beamformers under a sum power constraint, i.e., that produce Pareto efficient SINR tuples corresponding to points on the boundary of the hardening inner bound on the ergodic achievable rate region under a sum power constraint. In general, the structure of Pareto optimal downlink beamformers under a sum power constraints can be investigated by considering the following downlink sum power minimization problem
\begin{equation}\label{eq:yates_hard}
\begin{aligned}
\underset{\substack{(\forall k\in \set{K})~\vec{v}_k\in \set{V}_k}}{\text{minimize}} \quad & \sum_{k=1}^K\E[\|\vec{v}_k\|^2]\\
\text{subject to } \; \quad & (\forall k \in \set{K})~\mathsf{SINR}_k^{\mathsf{hard}}(\vec{v}_1,\ldots,\vec{v}_K) \geq \gamma_k.
\end{aligned}
\end{equation}
For example, one could assume that $(\gamma_1,\ldots,\gamma_K)$ is a sum-rate (Pareto) optimal SINR tuple. Although an efficient algorithm for determining such $(\gamma_1,\ldots,\gamma_K)$ is not available, studying the above problem can reveal the structure of sum-rate optimal downlink beamformers.  
In fact, by Proposition~\ref{prop:duality_hard}, optimal downlink beamformers can be equivalently obtained (up to a known scaling factor) by solving
\begin{equation*}
\begin{aligned}
\underset{\substack{(\forall k\in \set{K})~\vec{v}_k\in \set{V}_k\\ \vec{p}\in \stdset{R}_{+}^K}}{\text{minimize}} \quad & \sum_{k=1}^Kp_k\\
\text{subject to }\; \quad & (\forall k \in \set{K})~\mathsf{SINR}_k^{\mathsf{UatF}}(\vec{v}_k,\vec{p}) \geq \gamma_k,
\end{aligned}
\end{equation*}
which, due to the decoupled constraints in $(\vec{v}_1,\ldots,\vec{v}_K)$, can be further rewritten as \cite{miretti2023fixed}
\begin{equation*}
\begin{aligned}
\underset{\substack{\vec{p}\in \stdset{R}_{+}^K}}{\text{minimize}} \quad & \sum_{k=1}^Kp_k\\
\text{subject to} \quad & (\forall k \in \set{K})~\sup_{\vec{v}_k\in \set{V}_k}\mathsf{SINR}_k^{\mathsf{UatF}}(\vec{v}_k,\vec{p}) \geq \gamma_k.
\end{aligned}
\end{equation*}
We further recall that, by the results in Section~\ref{sec:uplink_UatF}, optimal beamformers attaining the above supremum can be obtained as solutions to MSE minimization problems in \eqref{eq:MSE_prob}.
Therefore, for the hardening inner bound under a sum power constraint, the main insights on joint optimal channel estimation and beamforming design from Section~\ref{sec:main} directly apply also to the downlink case. In particular, coming back to the sum-rate maximization example, we notice that optimal downlink beamformers under a sum power constraint take the form of \eqref{eq:CMMSE}, \eqref{eq:TMMSE}, \eqref{eq:LTMMSE}, or \eqref{eq:MTMMSE}, depending on the information constraints. The only remaining (albeit challenging) step toward a complete solution is the tuning of the $K$ nonnegative parameters $\vec{p}$ such that the sum rate is maximized.

\subsection{Coherent decoding lower bound}\label{sec:duality_coh}
Throughout this section, we assume genie-aided availability of some common decoder side information $U_k=U$ at each user $k \in \set{K}$, and that this information is also fully known at the transmit side for beamforming, i.e., $(\forall k \in \set{K})$ $(\forall l \in \set{L}_k)$ $S_{l,k} = (S_{l,k}',U)$. We point out that this is done only for optimization purposes, as similar performance can be achieved in practice with much more limited information. For example, note that practical decoders only require knowledge of the effective channel coefficients of the corresponding users. For every $k\in \set{K}$ and $(\vec{v}_1,\ldots,\vec{v}_K)\in \set{V}_1\times \ldots \times \set{V}_K$, we then let $\mathsf{SINR}_k^{\mathsf{coh,dl}}(\vec{v}_1,\ldots,\vec{v}_K)$  be the instantaneous SINR in \eqref{eq:coh} specialized to the downlink case, where the dependency on the beamformers is made explicit (we omit the dependency on $U$ for simplicity). Similarly, for every $k\in \set{K}$, $\vec{v}_k\in \set{V}_k$, and $\vec{p}\in \set{P}_U^K$, we let $\mathsf{SINR}_k^{\mathsf{coh,ul}}(\vec{v}_k,\vec{p})$ be the instantaneous SINR in \eqref{eq:coh} specialized to the uplink case. Given some instantaneous SINR constraints $(\gamma_1,\ldots,\gamma_K)\in \set{P}^K_{U}$, we consider the following downlink feasibility problem 
\begin{equation}\label{prob:feas_dl}
\begin{aligned}
\text{find} ~ & (\vec{v}_1,\ldots,\vec{v}_K,\vec{p})\in \set{V}_1\times \ldots \times \set{V}_K´\times \set{P}_U^K \\
\text{such that} ~ & (\forall k \in \set{K})~\mathsf{SINR}_k^{\mathsf{coh,dl}}(\sqrt{p_1}\vec{v}_1,\ldots,\sqrt{p_K}\vec{v}_K) = \gamma_k,
\end{aligned}
\end{equation}
and the uplink feasibility problem
\begin{equation}\label{prob:feas_ul}
\begin{aligned}
\text{find} ~ & (\vec{v}_1,\ldots,\vec{v}_K,\vec{p})\in \set{V}_1\times \ldots \times \set{V}_K\times \set{P}_U^K\\
\text{such that} ~ & (\forall k \in \set{K})~\mathsf{SINR}_k^{\mathsf{coh,ul}}(\vec{v}_k,\vec{p}) = \gamma_k.
\end{aligned}
\end{equation}
Note that the above stochastic feasibility problems consider almost sure equalities between random variables, which are all functions of $U$. We then have the following stochastic version of the uplink-downlink duality principle discussed in the previous section.

\begin{proposition}\label{prop:duality_coh}
Problem~\eqref{prob:feas_dl} admits a solution $(\vec{v}_1^\star,\ldots,\vec{v}_K^\star,\vec{p}^\mathsf{dl})$ if and only if Problem~\eqref{prob:feas_ul} admits a solution $(\vec{v}_1^{\star},\ldots,\vec{v}_K^{\star},\vec{p}^\mathsf{ul})$, for the same beamformers $(\vec{v}_1^\star,\ldots,\vec{v}_K^\star)$ ($\vec{p}^\mathsf{ul}$ and $\vec{p}^\mathsf{dl}$ can be different). Moreover, the power vector $\vec{p}^{\mathsf{dl}}$ (resp. $\vec{p}^{\mathsf{ul}}$) solves the system $(\vec{D}-\vec{\Gamma}\vec{B})\vec{p}^{\mathsf{dl}} = \vec{\Gamma}\vec{1}$ (resp. $(\vec{D}-\vec{\Gamma}\vec{B}^\T)\vec{p}^{\mathsf{ul}} = \vec{\Gamma}\vec{\Sigma}\vec{1}$) parametrized by $\vec{\Gamma} \eqdef \mathrm{diag}\left(\gamma_1,\ldots,\gamma_K\right)$, $\vec{D}\eqdef \mathrm{diag}\left(|\E[\vec{h}^\herm_1\vec{v}^{\star}_1|U]|^2,\ldots,|\E[\vec{h}^\herm_K\vec{v}^{\star}_K|U]|^2\right)$, $\vec{\Sigma}\eqdef \mathrm{diag}\left(\E[\|\vec{v}^{\star}_1\|^2|U],\ldots,\E[\|\vec{v}^{\star}_K\|^2|U]\right)$, and
\begin{align*}
\vec{B}&\eqdef \begin{bmatrix}
\E[|\vec{h}_1^\herm\vec{v}_1^{\star}|^2|U] & \ldots & \E[|\vec{h}_1^\herm\vec{v}^{\star}_K|U|^2] \\
\vdots & \ddots & \vdots \\
\E[|\vec{h}_K^\herm\vec{v}^{\star}_1|^2|U] & \ldots & \E[|\vec{h}_K^\herm\vec{v}^{\star}_K|U|^2]
\end{bmatrix}-\vec{D}.
\end{align*} 
Furthermore, the two solutions satisfy $\vec{1}^\T \vec{p}^\mathsf{ul}= \vec{1}^\T \vec{\Sigma}\vec{p}^\mathsf{dl}$.
\end{proposition}
\begin{proof}
The proof follows by extending the same arguments as in the proof of Proposition~\ref{prop:duality_hard} to a stochastic setting. 
\end{proof}

Similar to what discussed in Section~\ref{sec:duality_hard}, the above proposition can be used to study the structure of Pareto optimal downlink beamformers with respect to the coherent decoding inner bound on the ergodic achievable rate region under a sum power constraint. In particular, given minimum ergodic rate constraints $(r_1,\ldots,r_K)\in \stdset{R}^K_+$, the optimization problem
\begin{equation}\label{prob:QoS_coh}
\begin{aligned}
\underset{\substack{(\forall k\in \set{K})~\vec{v}_k\in \set{V}_k}}{\text{minimize}} \quad & \sum_{k=1}^K\E[\|\vec{v}_k\|^2]\\
\text{subject to } \; \quad & (\forall k \in \set{K})~R_k^{\mathsf{coh,dl}}(\vec{v}_1,\ldots,\vec{v}_K) \geq r_k,
\end{aligned}
\end{equation}
$R_k^{\mathsf{coh,dl}}(\vec{v}_1,\ldots,\vec{v}_K)\eqdef \E[\log(1+\mathsf{SINR}_k^{\mathsf{coh,dl}}(\vec{v}_1,\ldots,\vec{v}_K) )]$, can be converted using Proposition~\ref{prop:duality_coh} to
\begin{equation*}
\begin{aligned}
\underset{\substack{(\forall k\in \set{K})~\vec{v}_k\in \set{V}_k\\ \vec{p}\in \set{P}_{U}^K}}{\text{minimize}} \quad & \sum_{k=1}^K\E[p_k]\\
\text{subject to } \; \quad & (\forall k \in \set{K})~R_k^{\mathsf{coh,ul}}(\vec{v}_k,\vec{p}) \geq r_k,
\end{aligned}
\end{equation*}
$R_k^{\mathsf{coh,ul}}(\vec{v}_k,\vec{p})\eqdef \E[\log(1+\mathsf{SINR}_k^{\mathsf{coh,ul}}(\vec{v}_k,\vec{p}) )]$. For the above conversion to hold, we remark the crucial assumption on the availability of the side information $U$ for beamforming, which ensures that applying a power scaling factor $p_k\in \set{P}_U^K$ to $\vec{v}_k \in \set{V}_k$ still produces a beamformer $\sqrt{p_k}\vec{v}_k\in \set{V}_k$ for each $k\in \set{K}$.
Due to the decoupled constraints in $(\vec{v}_1,\ldots,\vec{v}_K)$, we then obtain that Pareto optimal downlink beamformers are given (up to a known scaling factor) by uplink beamformers attaining $\sup_{\vec{v}_k \in \set{V}_k} R_k^{\mathsf{coh},\mathsf{ul}}(\vec{v}_k,\vec{p})$ for some $\vec{p}\in \set{P}_U^K$. Furthermore, by Corollary~\ref{cor:coh} and by recalling that $U$ is common information, optimal beamformers can be obtained as solutions to MSE minimization problems in \eqref{eq:MSE_prob}. Therefore, for the coherent decoding inner bound under a sum power constraint, and under the given assumption on the availability of some common information $U$ for both decoding and beamforming, the main insights on joint optimal channel estimation and beamforming from Section~\ref{sec:main} directly apply also to the downlink case.

We conclude this section by pointing out that the above discussion also applies to the case of instantaneous sum power minimization, i.e., to the stochastic optimization problem
\begin{equation}\label{prob:QoS_coh_inst}
\begin{aligned}
\underset{\substack{(\forall k\in \set{K})~\vec{v}_k\in \set{V}_k}}{\text{minimize}} \quad & \sum_{k=1}^K\E[\|\vec{v}_k\|^2|U]\\
\text{subject to } \; \quad & (\forall k \in \set{K})~\mathsf{SINR}_k^{\mathsf{coh,dl}}(\vec{v}_1,\ldots,\vec{v}_K) \geq \gamma_k,
\end{aligned}
\end{equation}
where $(\gamma_1,\ldots,\gamma_K)\in \set{P}_U^K$ is a given tuple of instantaneous SINR constraints, and where the term $\textit{minimize}$ should be interpreted as finding the essential infimum of the set of feasible random variables $\E[\|\vec{v}_k\|^2|U]$ satisfying the instantaneous SINR constraints almost surely. Note that, if $(\gamma_1,\ldots,\gamma_K)$ is given by the instantaneous SINRs attained by a solution $(\vec{v}_1^\star,\ldots,\vec{v}_K^\star)$ to Problem~\eqref{prob:QoS_coh}, then Problem~\eqref{prob:QoS_coh_inst} is also solved by $(\vec{v}_1^\star,\ldots,\vec{v}_K^\star)$. However, Problem~\eqref{prob:QoS_coh_inst} and the uplink-downlink duality result in Proposition~\ref{prop:duality_coh} can be used to study the optimal solution structure for different SINR tuples $(\gamma_1,\ldots,\gamma_K)\in \set{P}_U^K$, e.g., corresponding to points on the (Pareto) boundary of the instantaneous rate region under an instantaneous sum power constraint.

\subsection{Extension to per-transmitter power constraints}\label{sec:duality_perTX}
Considering a sum power constraint is typically enough for studying multi-cell massive MIMO system models. Furthermore, due to its tractability, it is also an excellent proxy for studying cell-free massive MIMO system models. This is because scaling optimal downlink beamformers derived under a sum power constraint, in order to meet more realistic per-transmitter power constraints, typically incurs minor SNR losses in the interference limited regime. Nevertheless, in this section we show that the above insights on joint optimal channel estimation and beamforming under a sum power constraint can also be rigorously applied to the case of per-transmitter power constraints. 

Following the approach in \cite{miretti2024duality}, we focus on the following modified version of Problem~\eqref{eq:yates_hard}
\begin{equation}\label{eq:hard_pertx}
\begin{aligned}
\underset{\substack{(\forall k\in \set{K})~\vec{v}_k\in \set{V}_k}}{\text{minimize}} \quad & \sum_{k=1}^K\E[\|\vec{v}_k\|^2]\\
\text{subject to } \; \quad & (\forall k \in \set{K})~\mathsf{SINR}_k^{\mathsf{hard}}(\vec{v}_1,\ldots,\vec{v}_K) \geq \gamma_k\\
& (\forall l \in \set{L})~\sum_{k=1}^K\E[\|\vec{v}_{l,k}\|^2]\leq P_l,
\end{aligned}
\end{equation}
where $(P_1,\ldots,P_L)\in \stdset{R}_{++}^L$ is a given tuple of per-transmitter power constraints.
We then consider the following augmented Lagrangian formulation
\begin{equation}\label{eq:hard_pertx_aug}
\begin{aligned}
\underset{\substack{(\forall k\in \set{K})~\vec{v}_k\in \set{V}_k}}{\text{minimize}} \quad & \sum_{k=1}\E[\|\vec{v}_k\|_{\vec{1}+\vec{\lambda}}^2]\\
\text{subject to } \; \quad & (\forall k \in \set{K})~\mathsf{SINR}_k^{\mathsf{hard}}(\vec{v}_1,\ldots,\vec{v}_K) \geq \gamma_k,
\end{aligned}
\end{equation}
where $\vec{\lambda}=(\lambda_1,\ldots,\lambda_L)\in \stdset{R}_{+}^L$ are Lagrangian multipliers associated with the power constraints, and where we define $(\forall k \in \set{K})(\forall \vec{v}_k\in \set{V}_k)$ $\|\vec{v}_k\|_{\vec{1}+\vec{\lambda}}^2 \eqdef \sum_{l=1}^L(1+\lambda_l)\|\vec{v}_{l,k}\|^2$. The following proposition connects the above two problems.
\begin{proposition}\label{prop:hard_pertx_aug}
If Problem~\eqref{eq:hard_pertx} is feasible, then there exists $\vec{\lambda}\in \stdset{R}_{+}^L$ such that a solution to Problem~\eqref{eq:hard_pertx} is given by a solution to Problem~\eqref{eq:hard_pertx_aug}. 
\end{proposition}
\begin{proof}
The proof is given in \cite{miretti2024duality} for the special case of strict feasibility, and by further restricting $\set{V}_k$ to elements with finite second-order moment without loss of generality. In Appendix~\ref{app:proof_hard_pertx_aug}, we show how to modify the proof in \cite{miretti2024duality} to cover the more general feasibility assumption.
\end{proof}

We notice the similarity between Problem~\eqref{eq:hard_pertx_aug} and Problem~\eqref{eq:yates_hard}. Hence, we expect that Problem~\eqref{eq:hard_pertx_aug} can be similarly converted into a dual uplink problem. In fact, as observed in \cite{miretti2024duality}, for all $\vec{\lambda}\in \stdset{R}_+^L$, a solution to Problem~\eqref{eq:hard_pertx_aug} can be equivalently obtained (up to a known scaling factor) by solving 
\begin{equation*}\label{eq:Uatf_pertx}
\begin{aligned}
\underset{\substack{(\forall k\in \set{K})~\vec{v}_k\in \set{V}_k\\ \vec{p}\in \stdset{R}_{+}^K}}{\text{minimize}} \quad & \sum_{k=1}^Kp_k\\
\text{subject to }\; \quad & (\forall k \in \set{K})~\widetilde{\mathsf{SINR}}_k^{\mathsf{UatF}}(\vec{v}_k,\vec{p},\vec{\lambda}) \geq \gamma_k,
\end{aligned}
\end{equation*}
where $(\forall k \in \set{K})$ $(\forall \vec{v}_k \in\set{V}_k)$ $(\forall \vec{p}\in \stdset{R}_+^K)$ $(\forall \vec{\lambda}\in \stdset{R}_+^L)$ $ \widetilde{\mathsf{SINR}}_k^{\mathsf{UatF}}(\vec{v}_k,\vec{p},\vec{\lambda})\eqdef $
\begin{equation*}
\dfrac{p_k|\E[\vec{h}_k^\herm\vec{v}_k]|^2}{p_k\V(\vec{h}_k^\herm\vec{v}_k)+\sum_{j\in \set{K}\backslash \{k\}}p_j\E[|\vec{h}_j^\herm\vec{v}_k|^2]+\E[\|\vec{v}_k\|_{\vec{1}+\vec{\lambda}}^2]}
\end{equation*}
is a modified version of $\mathsf{SINR}_k^{\mathsf{UatF}}(\vec{v}_k,\vec{p})$ obtained by replacing the noise covariance $\vec{I}_M$ in \eqref{eq:ul} with the augmented noise covariance $\vec{I}_M+\vec{\Lambda} \eqdef \vec{I}_M + \mathrm{diag}(\lambda_1\vec{I}_N,\ldots,\lambda_L\vec{I}_N)$. This equivalence readily follows  by a minor modification of the uplink-downlink duality principle in Proposition~\ref{prop:duality_hard} obtained by replacing $\mathsf{SINR}_k^{\mathsf{UatF}}$ with $\widetilde{\mathsf{SINR}}_k^{\mathsf{UatF}}$ in Problem~\eqref{prob:feas_UatF}, and $\vec{\Sigma}=\mathrm{diag}\left(\E[\|\vec{v}_1^{\star}\|^2],\ldots,\E[\|\vec{v}_K^\star\|^2]\right)$ with $\vec{\Sigma}_{\vec{\lambda}}\eqdef \mathrm{diag}\left(\E[\|\vec{v}_1^{\star}\|_{\vec{1}+\vec{\lambda}}^2],\ldots,\E[\|\vec{v}_K^\star\|_{\vec{1}+\vec{\lambda}}^2]\right)$. Therefore, similar to the sum power constraint case, a solution to Problem~\eqref{eq:hard_pertx} can be equivalently obtained as scaled versions of uplink beamformers attaining $\sup_{\vec{v}_k\in \set{V}_k}\widetilde{\mathsf{SINR}}_k^{\mathsf{UatF}}(\vec{v}_k,\vec{p},\vec{\lambda})$ for some $(\vec{p},\vec{\lambda})\in \stdset{R}_+^K\times \stdset{R}_+^L$. 

Moreover, by using the same arguments on the relation between uplink SINR maximization and MSE minimization in Section~\ref{sec:uplink_UatF}, it can be readily shown that, for all  $(\vec{p},\vec{\lambda})\in \stdset{R}_+^K\times \stdset{R}_+^L$,  uplink beamformers attaining $\sup_{\vec{v}_k\in \set{V}_k}\widetilde{\mathsf{SINR}}_k^{\mathsf{UatF}}(\vec{v}_k,\vec{p},\vec{\lambda})$  can be obtained as solutions to modified MSE minimization problems
 $(\forall k \in \set{K})$
 
\begin{equation*}
\begin{aligned}
\underset{\vec{v}_k\in \set{V}_k}{\text{minimize}}~\widetilde{\mathsf{MSE}}_k(\vec{v}_k,\vec{p},\vec{
\lambda}),
\end{aligned}
\end{equation*} 
where $
\widetilde{\mathsf{MSE}}_k(\vec{v}_k,\vec{p},\vec{\lambda})= \E[\vec{v}_k^\herm(\vec{H}\vec{P}\vec{H}^\herm + \vec{I}_M+\vec{\Lambda})\vec{v}_k -2\sqrt{p_k}\Re(\vec{h}_k^\herm\vec{v}_k)+1]$. Finally, we observe that all the results in Section~\ref{sec:main} can be readily adapted to the above modified MSE expression. In particular, Proposition~\ref{prop:centr} on optimal centralized beamforming can be adapted by replacing $\vec{I}_M$ with $\vec{I}_M+ \vec{\Lambda}$ in the matrix inversion step in \eqref{eq:CMMSE}. Similarly, Proposition~\ref{prop:distr} on optimal distributed beamforming (as well as the given examples in Section~\ref{sec:examples}) can be adapted by replacing $\vec{I}_N$ with $\vec{I}_N+ \lambda_l \vec{I}_N$ in the matrix inversion step of each local MMSE beamformer in \eqref{eq:TMMSE}. Therefore, the main conclusions of Section~\ref{sec:duality_hard} can be extended to the case of per-transmitter power constraints, i.e., to Problem~\eqref{eq:hard_pertx}.

We conclude this section by pointing out that a similar result to Proposition~\ref{prop:hard_pertx_aug} can be established also for the coherent decoding lower bound, and in particular to a modified version of Problem~\eqref{prob:QoS_coh_inst} given by 
\begin{equation*}
\begin{aligned}
\underset{\substack{(\forall k\in \set{K})~\vec{v}_k\in \set{V}_k}}{\text{minimize}} \quad & \sum_{k=1}^K\E[\|\vec{v}_k\|^2|U]\\
\text{subject to } \; \quad & (\forall k \in \set{K})~\mathsf{SINR}_k^{\mathsf{coh,dl}}(\vec{v}_1,\ldots,\vec{v}_K) \geq \gamma_k \\
& (\forall l \in \set{L})~\sum_{k=1}^K\E[\|\vec{v}_{l,k}\|^2|U]\leq P_l,
\end{aligned}
\end{equation*}
where $(\gamma_1,\ldots,\gamma_K)\in \set{P}_U^K$ and $(P_1,\ldots,P_L)\in \set{P}_U^L$ are given tuples of instantaneous SINR and per-transmitter power constraints. Hence, following similar arguments as above, we obtain that the main conclusions of Section~\ref{sec:main} and Section~\ref{sec:duality_coh} can be extended to the case of per-transmitter power constraints. We omit the details due to space limitation.

\section{Applications}\label{sec:applications}
In this section we illustrate the application of our main results to common multi-cell and cell-free massive MIMO system models. In particular, after reviewing the canonical channel and pilot signaling model from the related literature, we show how to evaluate the solution to the MMSE joint channel estimation and beamforming problem \eqref{eq:MSE} in Section~\ref{sec:main} under different practical information constraints. We recall that, as discussed in Section~\ref{sec:uplink} and Section~\ref{sec:duality}, this problem is relevant for both uplink and downlink models. Furthermore, to be concrete, we demonstrate how this solution can be used to address uplink and downlink resource allocation problems that can be interpreted as joint beamforming, channel estimation, and power control problems. 

\subsection{Channel and pilot signaling model}\label{ssec:channel_pilot_model}
We consider a multi-user MIMO network composed by $L$ access points, each of them equipped with $N$ antennas. We let the channel between each access point $l\in \set{L}$ and user $k\in \set{K}$ be distributed as $\vec{h}_{l,k}\sim\CN(\vec{0},\vec{R}_{l,k})$ for some spatial covariance matrix $\vec{R}_{l,k}\in \stdset{C}^{N\times N}$, according to the spatially correlated Rayleigh fading model \cite[Chapter~2]{massivemimobook} \cite[Chapter~2]{demir2021}. Motivated by their geographical separation, we further assume that the channels of all user and access point pairs are mutually independent \cite[Chapter~2]{demir2021}. 
\begin{remark}
All the insights in this section can be readily extended to non-zero mean channels, e.g., to the spatially correlated Rician fading model for line-of-sight propagation \cite{ozdogan2019rician,noor2025wcnc}. The details are omitted, as they do not provide significantly new insights and result in cumbersome expressions.
\end{remark}

Following the usual uplink pilot signaling model in \cite[Chapter~3]{massivemimobook} \cite[Chapter~4]{demir2021}, which is essentially the uplink channel model in \eqref{eq:ul} applied to pilot transmission, we then assume that each access point $l\in \set{L}$ obtains $\tau_p$ noisy channel measurements of the type 
\begin{equation}\label{eq:ul_pilot}
\vec{Y}_{l}^{\mathsf{pilot}} = \sum_{k\in \set{K}} \sqrt{p_k^{\mathsf{pilot}}}\vec{h}_{l,k}\vec{\phi}_k^\herm + \vec{Z}^{\mathsf{pilot}}_{l},
\end{equation}
where $\vec{Z}^{\mathsf{pilot}}_l$ is a $N\times \tau_p$ additive white Gaussian matrix with $\CN(0,1)$ i.i.d. entries, and where $\vec{\phi}_k \in \stdset{C}^{\tau_p\times 1}$ is a known pilot sequence sent by user $k\in \set{K}$ with norm $\|\vec{\phi}_k\|^2=\tau_p$ and associated power control coefficient $p^{\mathsf{pilot}}_k \in \stdset{R}_+$. As customary, we assume that $\vec{\phi}_k$ is selected from the columns of a codebook of $\tau_p$ mutually orthogonal pilots $\vec{\Phi}\in \stdset{C}^{\tau_p\times \tau_p}$, where $\tau_p$ is typically much smaller than the coherence block length $\tau_c$ (i.e., the number of consecutive uplink channel uses in time and frequency over which the fading state can be assumed constant) and of the total number of users $K$, leading to the so-called pilot contamination effect. Given a user $k\in \set{K}$, we denote by $\set{K}_k\subseteq \set{K}$ the subset of users sharing the same pilot as user~$k$. According to the original massive MIMO paradigm, the above measurements are used for both uplink and downlink beamforming, leveraging the channel reciprocity property of time-division duplex systems.

To formulate the joint channel estimation and beamforming optimization problem \eqref{eq:MSE_prob}, we need to specify the underlying information constraints, i.e., the considered CSI $S_k=(S_{1,k},\ldots,S_{L,k})$ and the clusters $\set{L}_k$, as well as the power vector $\vec{p}$. In general, we assume that the CSI of each access point $l\in \set{L}$ is composed by at least the locally received uplink pilot signals $\vec{Y}_{l}^{\mathsf{pilot}}$ in \eqref{eq:ul_pilot}, and potentially by additional measurements obtained by sharing these signals over the fronthaul according to some specified rule. The cardinality of $\set{L}_k$ is then used to specify whether we consider a multi-cell or a cell-free system model. Several relevant examples are detailed in the following. 

\subsection{Multi-cell massive MIMO}\label{ssec:mMIMO}
In this section we focus on multi-cell massive MIMO networks, following closely the definition in \cite{massivemimobook} and other variations. For all $k\in \set{K}$, we let $\set{L}_k = \{l_k\}$ for some $l_k\in \set{L}$, i.e., we assume that each $k$th user is served by a single access point $l_k$. We further assume either $S_{l_k,k} = \vec{Y}_{l_k}^{\mathsf{pilot}}$ (no CSI sharing), i.e., that the serving access point only knows the locally measured pilots as in current cellular networks and in \cite{massivemimobook}, or $S_{l_k,k} = (\vec{Y}_{1}^{\mathsf{pilot}},\ldots,\vec{Y}_{L}^{\mathsf{pilot}})$ (full CSI sharing), i.e., that the access points share their measurements as in more advanced interference coordination schemes \cite{gesbert2010multicell}. The other variables $S_{l,k}$ for $l\neq l_k$ can be set arbitrarily, since the corresponding entries of $\vec{v}_k$ are anyway forced to zero. For the no CSI sharing case, we further assume $\vec{p}\in \stdset{R}^K_+$, i.e., that the power coefficients are deterministic as in \cite{massivemimobook}. For the full CSI sharing case, we also consider the more general case where $\vec{p}$ may be random and obtained as a function of $(\vec{Y}_{1}^{\mathsf{pilot}},\ldots,\vec{Y}_{L}^{\mathsf{pilot}})$, as in (rather impractical) fast adaptive power control schemes based on instantaneous CSI \cite{schubert2004solution}. We will give more details on this last condition in Section~\ref{ssec:maxmin_coh}.

Under the above assumptions, the solution to Problem~\eqref{eq:MSE_prob} for $k\in\set{K}$ is given by the local MMSE beamformer 
\begin{equation*}
(\forall l \in \set{L})~
    \vec{v}_{l,k} = \begin{cases}
        \vec{V}_{l,k}\vec{e}_k & \text{if } l = l_k, \\
        \vec{0} & \text{otherwise,}
    \end{cases}
\end{equation*}
where $\vec{V}_{l,k} = \Big(\hat{\vec{H}}_l^{(k)}\vec{P}\hat{\vec{H}}_l^{(k)\herm }+ \vec{\Psi}_l^{(k)}+\vec{I}_N  \Big)^{-1}\hat{\vec{H}}_l^{(k)}\vec{P}^{\frac{1}{2}}$ as in \eqref{eq:TMMSE}. This solution can be readily obtained from \eqref{eq:CMMSE} since, due to the multi-cell assumption $|\set{L}_k|=1$ which allows to arbitrarily set the other $S_{l,k}$ with $l\neq l_k$, both CSI sharing patterns can be mapped to a centralized CSI assumption as in Proposition~\ref{prop:centr} (team decision theory is not needed). To compute $\vec{V}_{l,k}$, we need to evaluate the MMSE estimate $\hat{\vec{H}}_l^{(k)} = \E[\vec{H}_l|\vec{Y}_{l}^{\mathsf{pilot}}]$ and the corresponding aggregate error covariance $\vec{\Psi}_l^{(k)}$. Each $j$th column of $\hat{\vec{H}}_l^{(k)}$, i.e., each $j$th user local channel estimate, can be rewritten as\footnote{Recall that we consider almost sure equalities between random variables.} $(\forall l \in \set{L})(\forall j \in \set{K})$
\begin{align*}
\E[\vec{h}_{l,j}|\vec{Y}_{l}^{\mathsf{pilot}}] = \E[\vec{h}_{l,j}|\vec{Y}_{l}^{\mathsf{pilot}}\vec{\Phi}] = \E[\vec{h}_{l,j}|\vec{y}_{l,j}^{\mathsf{pilot}}] =: \hat{\vec{h}}_{l,j},
\end{align*}
where the first equality follows since $\vec{\Phi}\in \stdset{C}^{\tau_p \times \tau_p}$ is a bijective map, and where the second equality follows by defining the decorrelated received uplink pilot signal $(\forall l \in \set{L})(\forall j \in \set{K})$
\begin{align*}\label{eq:ul_pilot_dec}
\vec{y}_{l,j}^{\mathsf{pilot}} &\eqdef \tfrac{1}{\sqrt{\tau_p}}\vec{Y}_{l}^{\mathsf{pilot}}\vec{\phi}_j\\
&= \sum_{i\in \set{K}_j} \sqrt{\tau_pp_i^{\mathsf{pilot}}}\vec{h}_{l,i} + \vec{z}^{\mathsf{pilot}}_{l,j}, \quad \vec{z}^{\mathsf{pilot}}_{l,j} \sim \CN(\vec{0}, \vec{I}_N),
\end{align*} 
and by noticing that $\vec{y}_{l,j}^{\mathsf{pilot}}$ is a sufficient statistics for $\vec{h}_{l,j}$ \cite[Chapter~4]{demir2021}. More precisely, we remark that $\vec{y}_{l,j}^{\mathsf{pilot}}$ is a column of $\vec{Y}_{l}^{\mathsf{pilot}}\vec{\Phi}$ and that, due to the independence of the channels of different users and the orthogonality of the pilots, the other columns of $\vec{Y}_{l}^{\mathsf{pilot}}\vec{\Phi}$ are independent of $\vec{h}_{l,j}$ and $\vec{y}_{l,j}^{\mathsf{pilot}}$. Then,  $\hat{\vec{H}}_l^{(k)}=[\hat{\vec{h}}_{l,1},\ldots,\ldots,\hat{\vec{h}}_{l,K}]$ is readily given by the well-known linear MMSE estimates $(\forall l \in \set{L})(\forall j\in \set{K})$ 
\begin{equation*}
\hat{\vec{h}}_{l,j} = \sqrt{\tau_pp_j^{\mathsf{pilot}}}\vec{R}_{l,j}\Bigg(\sum_{i\in \set{K}_j} \tau_pp_i^{\mathsf{pilot}}\vec{R}_{l,i}+\vec{I}_N \Bigg)^{-1}\vec{y}_{l,j}^{\mathsf{pilot}}.
\end{equation*}
Similarly, the error covariances in the definition of $\vec{\Psi}_l^{(k)}=\sum_{j\in \set{K}}p_j\E[\tilde{\vec{h}}_{l,j}^{(k)}\tilde{\vec{h}}_{l,j}^{(k)\herm}]$ are given by $(\forall l \in \set{L})(\forall j\in \set{K})$ 
\begin{align*}
&\E[\tilde{\vec{h}}_{l,j}^{(k)}\tilde{\vec{h}}_{l,j}^{(k)\herm}] = \E[(\hat{\vec{h}}_{l,j}-\vec{h}_{l,j})(\hat{\vec{h}}_{l,j}-\vec{h}_{l,j})^\herm] \\
&= \vec{R}_{l,j}-\tau_pp_j^{\mathsf{pilot}}\vec{R}_{l,j}\Bigg(\sum_{i\in \set{K}_j} \tau_pp_i^{\mathsf{pilot}}\vec{R}_{l,i}+\vec{I}_N \Bigg)^{-1}\vec{R}_{l,j}.
\end{align*}
Note that the index $k$ can be dropped in $\hat{\vec{H}}^{(k)}_l$,  $\vec{\Psi}^{(k)}_l$, and hence in $\vec{V}_{l,k}$. Therefore, we simply write $\hat{\vec{H}}_l$,  $\vec{\Psi}_l$, and $\vec{V}_l=\Big(\hat{\vec{H}}_l\vec{P}\hat{\vec{H}}_l^{\herm }+ \vec{\Psi}_l+\vec{I}_N  \Big)^{-1}\hat{\vec{H}}_l\vec{P}^{\frac{1}{2}}$ everywhere.  

\begin{remark}\label{rem:mMIMO_sharing}
If $\vec{p}$ is deterministic, then the solutions under no CSI sharing and full CSI sharing coincide, i.e., sharing the measured pilots is not beneficial. Furthermore, for random $\vec{p}$ and full CSI sharing, the solution can be computed by only sharing $\vec{p}$ instead of $(\vec{Y}_{1}^{\mathsf{pilot}},\ldots,\vec{Y}_{L}^{\mathsf{pilot}})$.
\end{remark}

\subsection{Cell-free massive MIMO}\label{ssec:cfMIMO}
We now focus on cell-free massive MIMO networks with arbitrary cooperation clusters and two levels of CSI sharing.
\subsubsection{No CSI sharing} For all $k\in \set{K}$, we let $|\set{L}_k|> 1$ and $(\forall l \in \set{L})$ $S_{l,k} = \vec{Y}_{l}^{\mathsf{pilot}}$. This corresponds to augmenting the multi-cell model with no CSI sharing covered in the previous section by allowing multiple access points to cooperate by sharing the data bearing signals, as in the distributed cell-free model defined in \cite{demir2021}. We also assume $\vec{p}\in \stdset{R}_+^K$, i.e., deterministic power coefficients. Under these assumptions, the solution to Problem~\eqref{eq:MSE_prob} for $k\in \set{K}$ is given by the local team MMSE beamformers \eqref{eq:LTMMSE} in Proposition~\ref{prop:local}, where the local MMSE beamforming stages $\vec{V}_{l,k}$ are computed using the same expressions for $\hat{\vec{H}}_l^{(k)}$ and  $\vec{\Psi}_l^{(k)}$ as in Section~\ref{ssec:mMIMO}. We recall that we can drop the index $k$ in $\vec{V}_{l,k}$, $\hat{\vec{H}}_l^{(k)}$, and  $\vec{\Psi}_l^{(k)}$, and hence simply write $\vec{V}_{l}$, $\hat{\vec{H}}_l$, and  $\vec{\Psi}_l$. The same observation applies to the parameters $\vec{\Pi}_{l,k} = \E[\vec{P}^{\frac{1}{2}}\hat{\vec{H}}_l^\herm\vec{V}_{l}]$ of the linear system of equations that needs to be solved (e.g., numerically) for computing the statistical beamforming stages $\vec{c}_{l,k}$. 

\subsubsection{Full CSI sharing}
For all $k\in \set{K}$, we let $|\set{L}_k|> 1$ and $(\forall l \in \set{L})$ $S_{l,k} = (\vec{Y}_{1}^{\mathsf{pilot}},\ldots,\vec{Y}_{L}^{\mathsf{pilot}})$. This corresponds to augmenting the multi-cell model by allowing multiple access points to cooperate by sharing both CSI and data bearing signals, as in the centralized cell-free model defined in \cite{demir2021}. Similar to the previous section, for the full CSI sharing case we may also allow $\vec{p}$ to be random and obtained as a function of $(\vec{Y}_{1}^{\mathsf{pilot}},\ldots,\vec{Y}_{L}^{\mathsf{pilot}})$. Under these assumption, the solution to Problem~\eqref{eq:MSE_prob} for $k\in \set{K}$ is given by the centralized MMSE beamformers \eqref{eq:CMMSE} in Proposition~\ref{prop:local}, where the  channel estimates and error covariances are given by
\begin{equation*}
\hat{\vec{H}}^{(k)} = \vec{C}_k\hat{\vec{H}}, \quad 
\vec{\Psi}^{(k)} = \vec{C}_k\vec{\Psi}\vec{C}_k,
\end{equation*}
\begin{equation*}
\hat{\vec{H}} \eqdef \begin{bmatrix}
\hat{\vec{H}}_1 \\ \vdots \\ \hat{\vec{H}}_L    
\end{bmatrix}, \quad 
\vec{\Psi} \eqdef \mathrm{diag}(\vec{\Psi}_1,\ldots,\vec{\Psi}_L),
\end{equation*}
where we use the same expressions for $\hat{\vec{H}}_l$ and  $\vec{\Psi}_l$ as in Section~\ref{ssec:mMIMO} and for the no CSI sharing case.
\begin{remark}\label{rem:cfMIMO_sharing}
If $\vec{p}$ is deterministic, then the solution under full CSI sharing can be computed by only sharing the received pilots (or the channel estimates) indexed by $l\in \set{L}_k$, as in the  centralized cell-free model in \cite{demir2021}. Furthermore, for random $\vec{p}$, the solution can be computed by further sharing $\vec{p}$ instead of the full tuple $(\vec{Y}_{1}^{\mathsf{pilot}},\ldots,\vec{Y}_{L}^{\mathsf{pilot}})$.
\end{remark}

\subsection{Long-term max-min fair resource allocation}\label{ssec:maxmin}
To illustrate the connection between the above MMSE-based solutions and the UatF/hardening bound in \eqref{eq:uatf}, we now consider examples of max-min fair resource allocation problems, which can be interpreted as joint beamforming, channel estimation, and power control problems. We first consider the uplink resource allocation problem
\begin{equation}\label{prob:maxmin_uatf}
\begin{aligned}
\underset{\substack{(\forall k\in \set{K})~\vec{v}_k\in \set{V}_k\\ \vec{p}\in \stdset{R}_{+}^K}}{\text{maximize}} \quad & \min_{k\in \set{K}} \; \mathsf{SINR}_k^{\mathsf{UatF}}(\vec{v}_k,\vec{p}) \\
\text{subject to }\; \quad & \|\vec{p}\|_{\infty} \leq P,
\end{aligned}
\end{equation}
where the information constraints $\set{V}_k$ are given by any of the three practical cases in Section~\ref{ssec:mMIMO} and Section~\ref{ssec:cfMIMO}, and $P\in \stdset{R}_{++}$ is a per-user uplink power constraint. This type of formulation has been termed a \textit{two-timescale} formulation in \cite{miretti2023fixed,miretti2024spawc}, since the optimization involves both  deterministic variables (the powers) that depend on large-scale fading only, and random variables (the beamformers) that adapt to small-scale fading. A solution to \eqref{prob:maxmin_uatf} can be obtained using the following normalized fixed-point algorithm from \cite{miretti2022joint,miretti2023fixed}:
\begin{algorithmic}[1]
    \Require $\vec{p} \in \stdset{R}_{++}^K$
    \Repeat
    \For{$k\in\set{K}$}
    \State $\vec{v}_k \gets  \arg\min_{\vec{v}_k\in\set{V}_k}\mathsf{MSE}_k(\vec{v}_k,\vec{p})$ 
    \EndFor
    \State $\vec{t} \gets \begin{bmatrix}
        \frac{p_1}{\mathsf{SINR}_1^{\mathsf{UatF}}(\vec{v}_1,\vec{p})} & \ldots & \frac{p_K}{\mathsf{SINR}_K^{\mathsf{UatF}}(\vec{v}_K,\vec{p})}
    \end{bmatrix}$
    \State $\vec{p}\gets \frac{P}{\|\vec{t}\|_{\infty}}\vec{t}$ 
    \Until{no significant progress is observed.}
\end{algorithmic}
The main step of the above algorithm is the MSE minimization step, which replaces a SINR maximization step following the MMSE-SINR relation in Proposition~\ref{prop:MMSE-SINR}, and which can be performed as discussed in the previous sections depending on the choice of the information constraint $\set{V}_k$. The SINRs in step~5 can be evaluated using the aforementioned relation, or the equivalent expression $(\forall k \in \set{K})$ $\mathsf{SINR}_k^{\mathsf{UatF}}(\vec{v}_k,\vec{p}) =$
\begin{equation*}
     \dfrac{p_k|\E[\hat{\vec{h}}_{k}^\herm\vec{v}_k]|^2 }{p_k\V(\hat{\vec{h}}_{k}^\herm\vec{v}_k)+\underset{j\in \set{K}\backslash \{k\}}{\sum} p_j\E[|\hat{\vec{h}}_{j}^\herm\vec{v}_k|^2]+\E[\vec{v}_k^\herm(\vec{\Psi} +\vec{I})\vec{v}_k]},
\end{equation*}
where $(\forall j \in \set{K})$ $\hat{\vec{h}}_j^\herm \eqdef [\hat{\vec{h}}_{1,j}^\herm,\ldots,\hat{\vec{h}}_{L,j}^\herm]$, which can be derived by manipulating the expectations in \eqref{eq:uatf} using the law of total expectation given $(\vec{Y}_{1}^{\mathsf{pilot}},\ldots,\vec{Y}_{L}^{\mathsf{pilot}})$, and the fact that $\vec{v}_k$ is a function of $(\vec{Y}_{1}^{\mathsf{pilot}},\ldots,\vec{Y}_{L}^{\mathsf{pilot}})$. 
Note that a very similar algorithm is given in \cite{miretti2024spawc} for the case of weighted sum-rate maximization, although with local optimality guarantees only. 

We then consider the similar downlink problem
\begin{equation}\label{prob:maxmin_hard}
\begin{aligned}
\underset{\substack{(\forall k\in \set{K})~\vec{v}_k\in \set{V}_k}}{\text{maximize}} \quad & \min_{k\in \set{K}} \; \mathsf{SINR}_k^{\mathsf{hard}}(\vec{v}_1,\ldots,\vec{v}_K) \\
\text{subject to }\; \quad & \sum_{k=1}^K\E[\|\vec{v}_k\|^2] \leq P,
\end{aligned}
\end{equation}
where $P\in \stdset{R}_{++}$ is a network-wide sum power constraint. Recalling the uplink-downlink duality principle in Proposition~\ref{prop:duality_hard},  Problem~\eqref{prob:maxmin_hard} can be transformed into a completely equivalent dual uplink problem of the form in \eqref{prob:maxmin_uatf}, with the $l_\infty$ norm replaced by the $l_1$ norm in the power constraint. The solution to this dual uplink problem can be then obtained using the same algorithm as above, by replacing the $l_{\infty}$ norm with the $l_1$ norm in step 6. Note that the same argument can be applied to obtain (locally optimal) solutions to downlink weighted sum-rate maximization problems under a sum power constraint using a variation of the algorithm in \cite{miretti2024spawc}.

A more realistic yet involved downlink version of \eqref{prob:maxmin_hard} with per-transmitter power constraints can also be addressed using known algorithms and the results in this work. In particular, we can apply a conventional bisection loop to solve max-min problems via a sequence of feasibility problems of the type in \eqref{eq:hard_pertx}. Recalling Proposition~\ref{prop:hard_pertx_aug} and the associated discussion, the corresponding augmented Lagrangian formulation~\eqref{eq:hard_pertx_aug} can be then transformed into a dual uplink feasibility problem, and solved using the fixed-point algorithm \cite{miretti2024duality,miretti2023fixed}
\begin{algorithmic}[1]
    \Require $\vec{p} \in \stdset{R}_{++}^K$, $\vec{\lambda}\in \stdset{R}_{+}^L$
    \Repeat
    \For{$k\in\set{K}$}
    \State $\vec{v}_k \gets  \arg\min_{\vec{v}_k\in\set{V}_k}\widetilde{\mathsf{MSE}}_k(\vec{v}_k,\vec{p},\vec{
\lambda})$ 
    \EndFor
    \State $\vec{p} \gets \begin{bmatrix}
        \frac{p_1}{\widetilde{\mathsf{SINR}}_1^{\mathsf{UatF}}(\vec{v}_1,\vec{p},\vec{\lambda})} & \ldots & \frac{p_K}{\widetilde{\mathsf{SINR}}_K^{\mathsf{UatF}}(\vec{v}_K,\vec{p},\vec{\lambda})}
    \end{bmatrix}$ 
    \Until{no significant progress is observed.}
\end{algorithmic}
The Lagrangian multipliers $\vec{\lambda}$ can be tuned using an outer projected gradient loop as in \cite[Proposition~5]{miretti2024duality}. We refer to \cite{miretti2024duality} for details on the overall algorithm for solving \eqref{eq:hard_pertx}. The relevant novelty in the context of joint channel estimation and beamforming lies in the MSE minimization step 3, which can be performed using a minor variation of the schemes in Section~\ref{ssec:mMIMO} and Section~\ref{ssec:cfMIMO} obtained by redefining $\vec{V}_l\eqdef\Big(\hat{\vec{H}}_l\vec{P}\hat{\vec{H}}_l^{\herm }+ \vec{\Psi}_l+\lambda_l\vec{I}_N  \Big)^{-1}\hat{\vec{H}}_l\vec{P}^{\frac{1}{2}}$, i.e., by considering a scaled noise covariance matrix for each access point.

\subsection{Short-term max-min fair resource allocation}\label{ssec:maxmin_coh}
In this section we illustrate the connection between the considered MMSE-based solutions and the coherent decoding lower bound in \eqref{eq:coh}. This also allow us to further elucidate the role of common information for beamforming, and of a random $\vec{p}$.
We consider the stochastic optimization problem
\begin{equation}\label{prob:maxmin_coh}
\begin{aligned}
\underset{\substack{(\forall k\in \set{K})~\vec{v}_k\in \set{V}_k\\ \vec{p}\in \set{P}_{U}^K}}{\text{maximize}} \quad & \min_{k\in \set{K}} \; \mathsf{SINR}_k^{\mathsf{coh,ul}}(\vec{v}_k,\vec{p}|U) \\
\text{subject to }\; \quad & \|\vec{p}\|_{\infty} \leq P,
\end{aligned}
\end{equation}
where $U\eqdef (\vec{Y}_{1}^{\mathsf{pilot}},\ldots,\vec{Y}_{L}^{\mathsf{pilot}})$, $P\in \stdset{R}_{++}$, and where the conditioning of the SINR on $U$ is made explicit. Notice that, in contrast to Problem~\eqref{prob:maxmin_hard}, $\vec{p}$ can be any function of $U$ and hence it can be random. The term \textit{maximize} should be interpreted as computing the essential supremum of the random variable $\min_{k\in \set{K}} \; \mathsf{SINR}_k^{\mathsf{coh,ul}}(\vec{v}_k,\vec{p}|U)$ while satisfying the power constraint almost surely. In practice, this means solving an instance of an instantaneous max-min SINR problem for each realization of $U=(\vec{Y}_{1}^{\mathsf{pilot}},\ldots,\vec{Y}_{L}^{\mathsf{pilot}})$. Taking aside the joint channel estimation aspect, this is essentially the approach followed by all \textit{short-term} methods in the classical MU-MIMO literature that solve an instantaneous joint beamforming and power control problem for each channel realization \cite{schubert2004solution}. For the case of full CSI sharing (in both multi-cell massive MIMO and cell-free massive MIMO models), Problem~\eqref{prob:maxmin_coh} can be solved using the following instantaneous version of the fixed-point algorithm for solving Problem~\eqref{prob:maxmin_uatf}, to be run disjointly for each realization $u\in \set{U}$ of $U$:
\begin{algorithmic}[1]
    \Require $u\in \set{U}$, $\vec{p} \in \stdset{R}_{++}^K$
    \Repeat
    \For{$k\in\set{K}$}
    \State $\vec{v}_k \gets  \arg\min_{\vec{v}_k\in\set{V}_k}\mathsf{MSE}_k(\vec{v}_k,\vec{p}|U=u)$ 
    \EndFor
    \State $\vec{t} \gets \begin{bmatrix}
        \frac{p_1}{\mathsf{SINR}_1^{\mathsf{coh,ul}}(\vec{v}_1,\vec{p}|U=u)} & \ldots & \frac{p_K}{\mathsf{SINR}_K^{\mathsf{coh,ul}}(\vec{v}_K,\vec{p}|U=u)}
    \end{bmatrix}$
    \State $\vec{p}\gets \frac{P}{\|\vec{t}\|_{\infty}}\vec{t}$ 
    \Until{no significant progress is observed.}
\end{algorithmic}
Specifically, by letting $(\vec{v}_1^{(u)},\ldots,\vec{v}_K^{(u)},\vec{p}^{(u)})$ be the limit of the sequence produced by the above algorithm for a given realization $u\in \set{U}$, then $(\vec{v}_1^{(U)},\ldots,\vec{v}_K^{(U)},\vec{p}^{(U)})$ solves Problem~\eqref{prob:maxmin_coh}. Note that, by construction, the above algorithm may produce a different $\vec{v}_k^{(u)}$ for each realization $u$ of $U$, which means that all entries of $\vec{v}_k^{(U)}$ may be a function of $U$. This is perfectly in line with the full CSI sharing assumption, since $\vec{v}_k^{(U)}\in \set{V}_k$, i.e.,  the information constraints are satisfied. In contrast, we notice that the above algorithm cannot be used to solve Problem~\eqref{prob:maxmin_coh} under no CSI sharing: in this case, while $\vec{v}_k^{(u)}\in \set{V}_k$, we have that $\vec{v}_k^{(U)}\notin \set{V}_k$ in general, i.e.,  $\vec{v}_k^{(U)}$ may violate the information constraints. A similar discussion applies to $\vec{p}^{(U)}$, which is a function of $U$ and hence potentially random. 

Furthermore, we remark that the full CSI sharing assumption not only ensures that the output of the above algorithm is feasible, but also that conditional SINR maximization in $\vec{v}_k$ can be converted to conditional and unconditional MSE minimization in $\vec{v}_k$, see Proposition~\ref{prop:MMSE-SINR-coh} and Proposition~\ref{prop:ess}. In particular, step 3 of the above algorithm can be performed by applying the same expressions in Section~\ref{ssec:mMIMO} and Section~\ref{ssec:cfMIMO} for the case of full CSI sharing and random $\vec{p}$. In addition, we notice that, under the full CSI sharing assumption, the SINR can be rewritten in the more familiar form $(\forall k \in \set{K})$
\begin{equation*}
\mathsf{SINR}_k^{\mathsf{coh,ul}}(\vec{v}_k,\vec{p}|U) = \dfrac{p_k|\hat{\vec{h}}_{k}^\herm\vec{v}_k|^2 }{\underset{j\in \set{K}\backslash \{k\}}{\sum}p_j|\hat{\vec{h}}_{j}^\herm\vec{v}_k|^2+\vec{v}_k^\herm(\vec{\Psi} +\vec{I})\vec{v}_k}.
\end{equation*}
As a side remark, we point out that equivalent optimal solutions can be obtained by replacing step 3 in the above algorithm with the direct maximization of the above SINR expression using the generalized Rayleigh quotient technique.  

We conclude this section by remarking that downlink versions of Problem~\eqref{prob:maxmin_coh} can be solved using the uplink-downlink duality principle in Section~\ref{sec:duality_coh} and similar arguments as in Section~\ref{ssec:maxmin}. Further details are omitted since beyond the scope of this work.

\section{Numerical examples}
\label{sec:simulations}
To further illustrate the practical relevance of our main results, this section presents a numerical evaluation of various sub-6 GHz wideband system deployments, including small-cell and cell-free architectures with different levels of cooperation, as described in Section~\ref{sec:applications}. A key feature of these numerical results is that, in many cases, we can guarantee that no alternative beamforming or channel estimation scheme, including joint schemes, can achieve better performance. We will restate the specific assumptions under which this property holds alongside each example. We emphasize that the focus of this section is primarily methodological, rather than to provide an exhaustive performance comparison of practical systems.

\subsection{Simulation setup}
We simulate an instance of the canonical network model described in Section~\ref{ssec:channel_pilot_model}, parametrized as follows. We assume $K=32$ users uniformly distributed within a squared service area of size $500\times 500~\text{m}^2$, and $L=16$ regularly spaced access points with $N=8$ antennas each. By neglecting for simplicity spatial correlation within the access points, we assume diagonal spatial covariances $\vec{R}_{l,k}= \beta_{l,k}\vec{I}_N$, where $\beta_{l,k}>0$ denotes the channel gain between access point $l$ and user $k$. We follow the same 3GPP-like path-loss model adopted in \cite{demir2021} for a $2$ GHz carrier frequency:
	\begin{equation*}
		\beta_{l,k} = -36.7 \log_{10}\left(D_{l,k}/1 \; \mathrm{m}\right) -30.5 + \Xi_{l,k} -\sigma^2 \quad \text{[dB]},
	\end{equation*}
where $D_{l,k}$ is the distance between access point $l$ and user $k$ including a difference in height of $10$ m, and $\Xi_{l,k}\sim \mathcal{N}(0,\rho^2)$ [dB] are shadow fading terms with deviation $\rho = 4$. The shadow fading is correlated as $\E[\Xi_{l,k}\Xi_{j,i}]=\rho^22^{-\frac{\delta_{k,i}}{9 \text{ [m]}}}$ for all $l=j$ and zero otherwise, where $\delta_{k,i}$ is the distance between user $k$ and user $i$. The noise power is $\sigma^2 = -174 + 10 \log_{10}(B) + F$ [dBm], where $B = 20$ MHz is the bandwidth, and $F = 7$ dB is the noise figure. 

Regarding uplink pilot signaling, we parametrize the model in \eqref{eq:ul_pilot} by assuming pilot sequences of length $\tau_p = 10$  and power $p^{\mathsf{pilot}}_k = 20$ dBm for all users $k\in \set{K}$. The length $\tau_p$ is chosen, in particular, to guarantee sufficient uplink channel estimation SNR for all users in the considered scenario. Pilots are assigned to users on a cell basis, following a simple round robin scheme in each cell. More specifically, we first cluster the users into $L$ cells, where each user $k$ is assigned to a cell $l$ corresponding to an access point with maximum gain $\beta_{l,k}$. Then, for each cell, the pilots are assigned sequentially to each user in the cell by cycling over the pilot indexes $\{1,\ldots,\tau_p\}$, starting from an arbitrary index. Note that, in the considered scenario, each cell is formed by $2 \ll \tau_p$ users on average, hence pilot contamination is not a major issue.

\begin{figure}[htp]
    \centering
    \begin{subfigure}{0.9\columnwidth}
        \centering        \includegraphics[width=\textwidth]{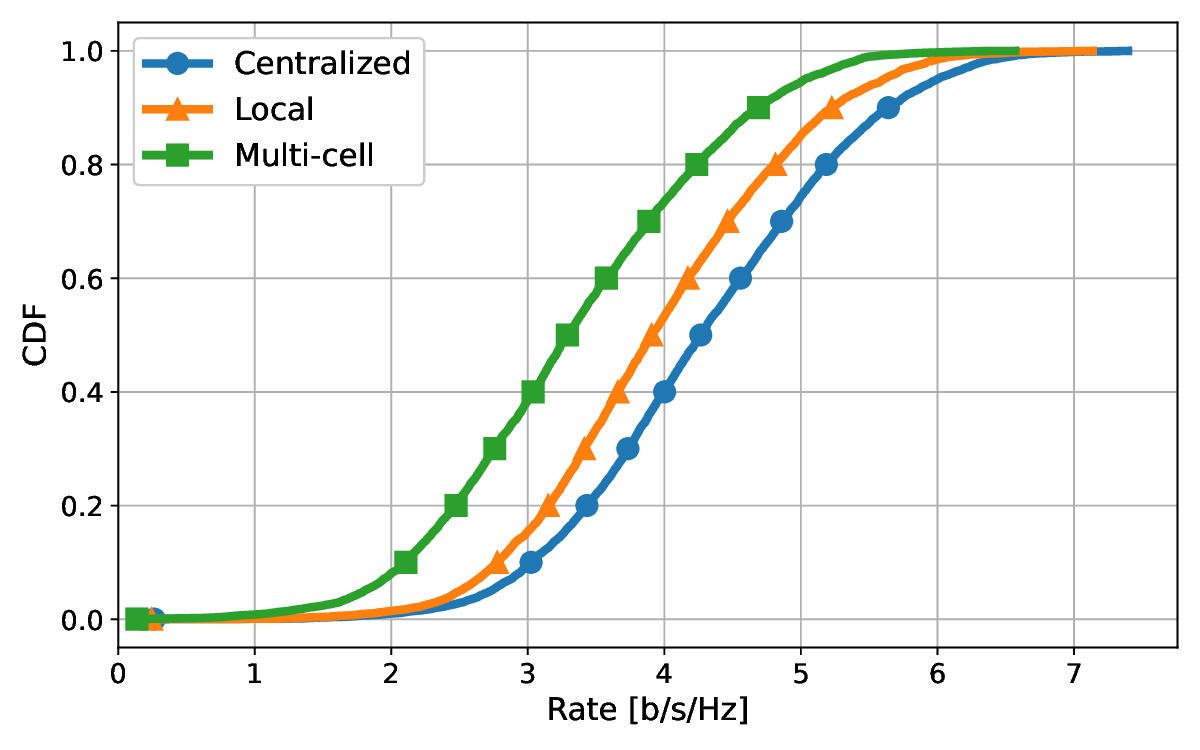}
        \caption{}
        \label{fig:fixed_uatf}
    \end{subfigure}
    \begin{subfigure}{0.9\columnwidth}
        \centering
    \includegraphics[width=\textwidth]{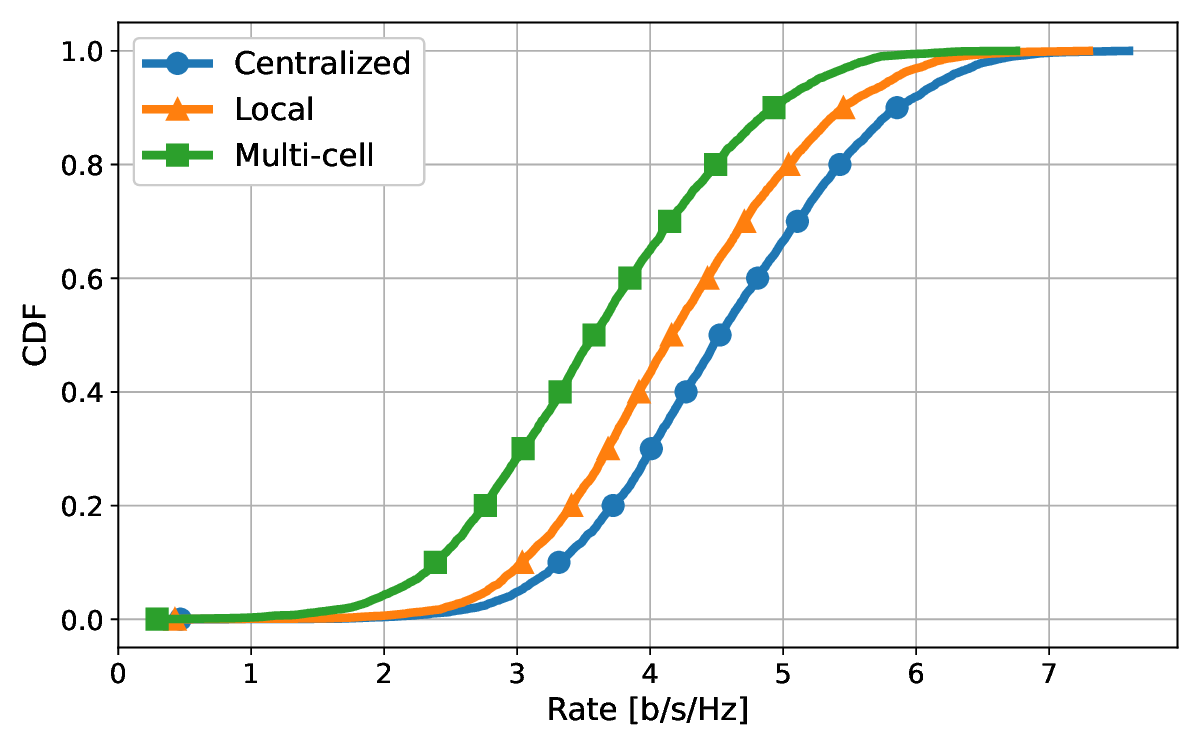}
        \caption{}
        \label{fig:fixed_coh}
    \end{subfigure}
    \begin{subfigure}{0.9\columnwidth}
        \centering
      \includegraphics[width=\textwidth]{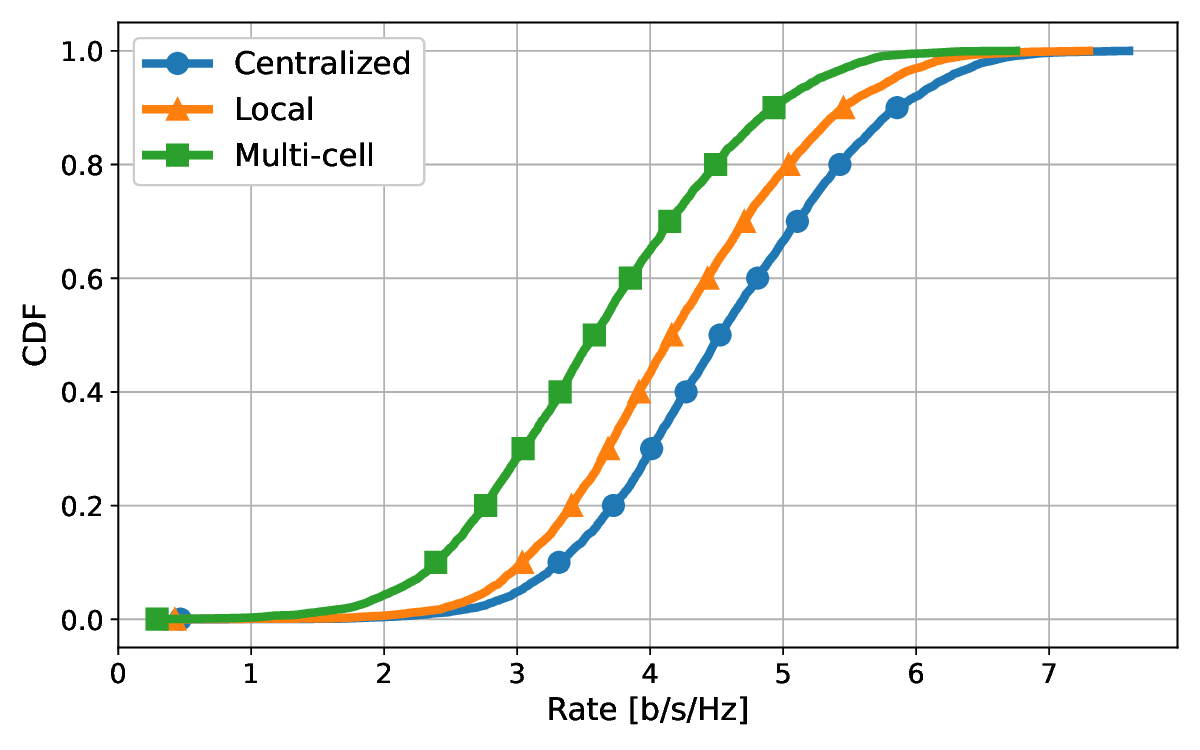}
        \caption{}
        \label{fig:fixed_oer}
    \end{subfigure}
    \caption{Empirical CDF of uplink ergodic rates in \eqref{eq:mutual_info} for different MSE-optimal joint beamforming and channel estimation schemes and a fractional power control policy. The uplink ergodic rates are approximated by using (a) the UatF lower bound in \eqref{eq:uatf}, (b) the coherent decoding lower bound \eqref{eq:coh}, and (c) the optimistic upper bound in \eqref{eq:oer}.} 
    \label{fig:fixed}
\end{figure}

\subsection{Performance of MSE-based uplink processing}
Figure~\ref{fig:fixed} reports the performance of the optimal solution to the MSE-based uplink joint channel estimation and beamforming problem in \eqref{eq:MSE_prob}, considering three different information constraints labeled as (i) \textit{multi-cell}, (ii) \textit{local}, and (iii) \textit{centralized}. For illustration purposes, we adopt a simple statistical fractional power control policy \cite[Eq.~(7.34)]{demir2021} $(\forall k \in \set{K})$ $p_k = \frac{(\sum_{l\in \set{L}_k}\beta_{l,k})^{-1}}{\max_{i\in \set{K}}(\sum_{l\in \set{L}_i}\beta_{l,i})^{-1}}P$ with  maximum per-user power $P = 20$ dBm, which targets approximate max-min fairness. 

Case (i) corresponds to the multi-cell massive MIMO model in Section~\ref{ssec:mMIMO}, where cells are formed based on the strongest channel gain, similarly to the pilot allocation scheme. In this case, since $\vec{p}$ is deterministic, the solutions under no CSI sharing and full CSI sharing coincide (see Remark~\ref{rem:mMIMO_sharing}). Cases (ii) and (iii) correspond to the two cell-free massive MIMO models described in Section~\ref{ssec:cfMIMO}, with no CSI sharing and full CSI sharing, respectively. In both cases, each user is served by its four strongest access points in terms of channel gain. Similarly to case (i), we note that in case (iii), restricting CSI sharing within the cluster of serving access points yields the same solution (see Remark~\ref{rem:cfMIMO_sharing}).

The performance of the above schemes is evaluated using the UatF lower bound \eqref{eq:uatf}, the coherent decoding lower bound \eqref{eq:coh} with decoder CSI $(\forall k \in \set{K})$ $U_k=(\vec{Y}_{1}^{\mathsf{pilot}},\ldots,\vec{Y}_{L}^{\mathsf{pilot}})$, and the upper bound \eqref{eq:oer} on the achievable ergodic rates \eqref{eq:mutual_info}. The figures show the empirical cumulative distribution function (CDF) of each performance metric over 200 i.i.d. user drops.

The numerical results can be interpreted as follows. First, under the considered power control policy, it follows from Section~\ref{sec:uplink_UatF} that Figure~\ref{fig:fixed_uatf} shows optimal rate profiles in terms of the UatF lower bound for all information constraints. Specifically, the UatF rates for each user cannot be further improved by using alternative channel estimation and beamforming schemes, including joint schemes. More precisely, the reported rate profiles are weakly Pareto optimal, i.e., they lie on the boundary of the achievable UatF rate region under the given information constraints and per-user power constraint $P$. Higher rates for some users can only be achieved by changing the power control policy, which typically results in rate degradation for others.

Furthermore, based on the results in Section~\ref{sec:duality_coh}, the same conclusion applies to the rate profiles in Figure~\ref{fig:fixed_coh} for the multi-cell and centralized cell-free massive MIMO models. This follows from the key assumption that, in these two models, the decoder CSI $U_k$ also serves as the common beamforming CSI, as per Definition~\ref{def:common_CSI}. (However, we recall that the same rates can be achieved with reduced CSI sharing, as previously discussed.) In contrast, this assumption does not hold in the case of the cell-free massive MIMO model with no CSI sharing. Therefore, the corresponding rate profile in Figure~\ref{fig:fixed_coh} may be suboptimal.

Lastly, we note that none of the considered MSE-based schemes are provably optimal in terms of optimistic ergodic rates (Figure~\ref{fig:fixed_oer}), as this bound assumes perfect decoder CSI, while the beamformers are constrained to be functions of noisy channel estimates. Nevertheless, we observe that Figure~\ref{fig:fixed_coh} and Figure~\ref{fig:fixed_oer} are essentially shifted versions of Figure~\ref{fig:fixed_uatf}, which suggests that the MSE criterion may serve as an excellent proxy for optimizing less tractable rate bounds. We report that this seems to be the case when the available CSI is of sufficiently high quality to effectively serve $K$ spatially multiplexed users, which is indeed the desired operating regime for well-designed massive MIMO systems. Our interpretation is that high interference cancellation capabilities can also be used to mitigate the self-interference artifacts of the UatF bound.

\begin{figure}[htp]
    \centering
    \begin{subfigure}{0.9\columnwidth}
        \centering
    \includegraphics[width=\textwidth]{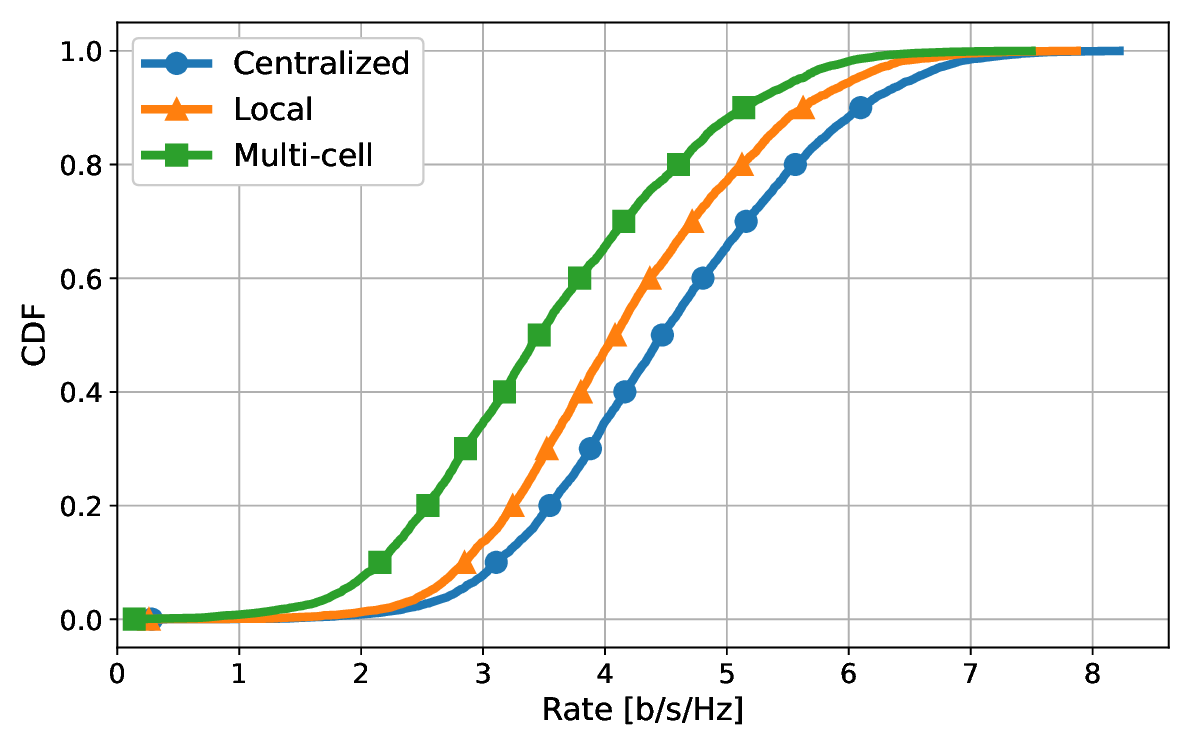}
        \caption{}
        \label{fig:fixed_sum_hard}
    \end{subfigure}
    \begin{subfigure}{0.9\columnwidth}
        \centering
    \includegraphics[width=\textwidth]{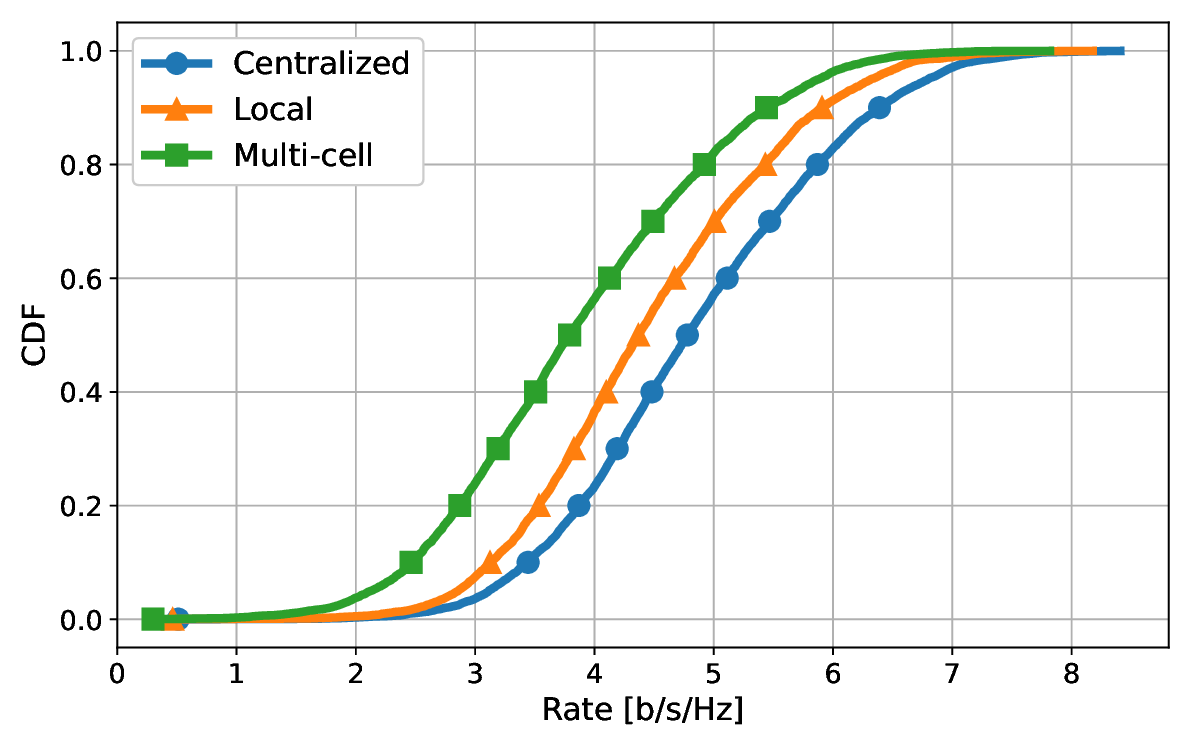}
        \caption{}
        \label{fig:fixed_sum_coh}
    \end{subfigure}
    \begin{subfigure}{0.9\columnwidth}
        \centering
    \includegraphics[width=\textwidth]{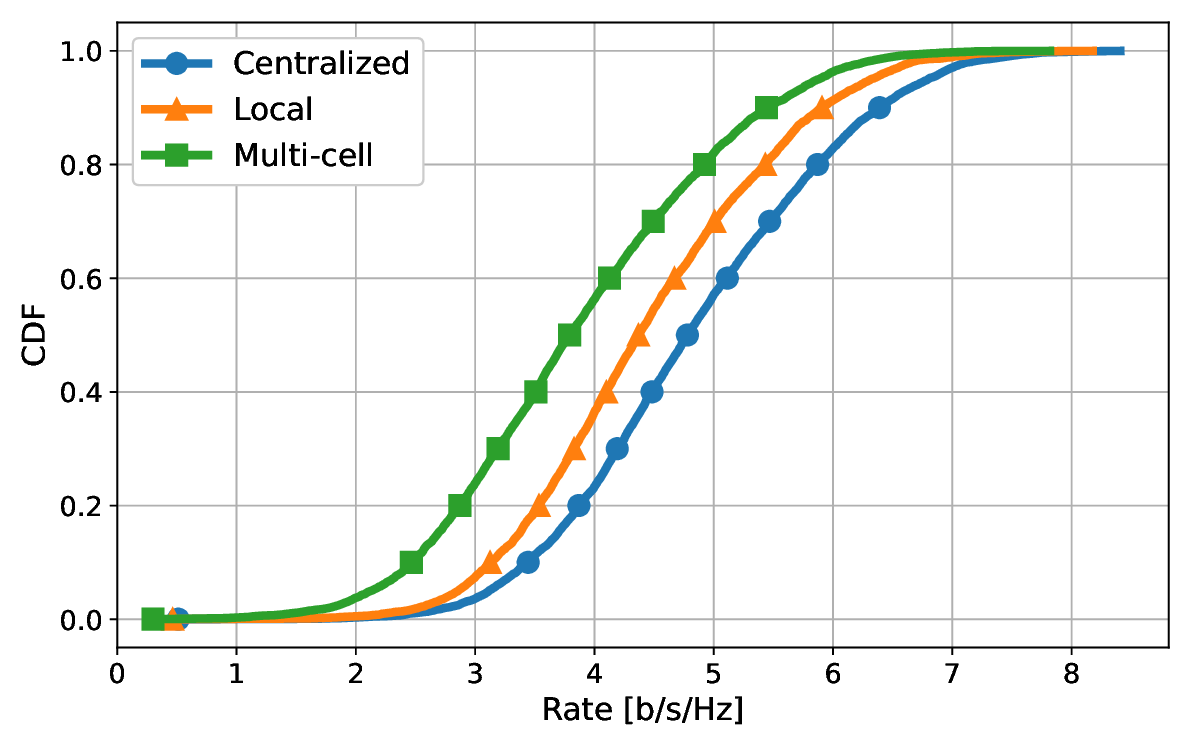}
        \caption{}
        \label{fig:fixed_sum_oer}
    \end{subfigure}
    \caption{Empirical CDF of downlink ergodic rates in \eqref{eq:mutual_info} for different MSE-optimal joint beamforming and channel estimation schemes and a fractional power control policy. The downlink ergodic rates are approximated by using the  dual of (a) the hardening lower bound in \eqref{eq:uatf}, (b) the coherent decoding lower bound \eqref{eq:coh}, and (c) the optimistic upper bound in \eqref{eq:oer}.}
    \label{fig:fixed_sum}
    \vspace{-0.5cm}
\end{figure}

\subsection{Performance of MSE-based downlink processing}
In this section, we focus on some applications of the uplink-duality principles presented in Section~\ref{sec:duality}, and illustrate how similar numerical experiments to those in the previous section can also be conducted to study the downlink case. In particular, Figure~\ref{fig:fixed_sum} reports the performance of the same MSE-based joint channel estimation and beamforming problem \eqref{eq:MSE_prob}, under the same information constraints, and using a slightly modified power control policy $(\forall k \in \set{K})$ $p_k = \frac{(\sum_{l\in \set{L}k}\beta{l,k})^{-1}}{\sum_{i\in \set{K}}(\sum_{l\in \set{L}i}\beta{l,i})^{-1}}KP$. The performance is measured using the same uplink ergodic rate bounds as in the previous section, which are now interpreted as virtual uplink rates over a dual uplink channel. The power control policy can be viewed as an approximate max-min fair policy for this dual uplink channel, subject to a sum power constraint of $KP$ across all users. We now discuss whether, and under what conditions, these dual uplink rates (which are purely mathematical constructs and may be entirely disconnected from the actual uplink rates) can also be interpreted as downlink rates.

As discussed in the previous section, the dual uplink rates in Figure~\ref{fig:fixed_sum_hard} are (Pareto) optimal. Furthermore, by Proposition~\ref{prop:duality_hard}, for every virtual uplink rate tuple $(R_1^{\mathsf{UatF}},\ldots,R_K^{\mathsf{UatF}})$, there exists a downlink rate tuple $(R_1^{\mathsf{hard}},\ldots,R_K^{\mathsf{hard}}) = (R_1^{\mathsf{UatF}},\ldots,R_K^{\mathsf{UatF}})$ that is achievable using the same beamformers (up to deterministic power scaling factors $\vec{p}^{\mathsf{dl}} \in \stdset{R}_+^K$) and the same total power. Therefore, the rate profiles in Figure~\ref{fig:fixed_sum_hard} are also (Pareto) optimal downlink rate profiles in terms of the hardening lower bound, under a sum power constraint $KP$, for all information constraints. Moreover, these rate profiles can serve as upper bounds for scenarios with more realistic per-access point power constraints. Lower bounds on the actual performance under such per-access point constraints can be obtained by applying heuristic power scaling factors to the sum-power optimal solutions and/or by leveraging the uplink-downlink duality principle under per-access point power constraints, as discussed in Section~\ref{sec:duality_perTX}. Additional details are omitted, as they are beyond the scope of this study.

The case of the coherent decoding lower bound is more complex, but still insightful. By Proposition~\ref{prop:duality_coh}, we know that the rate profiles in Figure~\ref{fig:fixed_sum_coh} can also be interpreted as downlink achievable rates in terms of the coherent decoding lower bound, under certain conditions discussed below. Specifically, for every virtual uplink rate tuple $(R_1^{\mathsf{coh,ul}},\ldots,R_K^{\mathsf{coh,ul}})$, there exists a downlink rate tuple $(R_1^{\mathsf{coh,dl}},\ldots,R_K^{\mathsf{coh,dl}}) = (R_1^{\mathsf{coh,ul}},\ldots,R_K^{\mathsf{coh,ul}})$ that is achievable assuming the same decoder CSI $(\forall k \in \set{K})$ $U_k = (\vec{Y}_{1}^{\mathsf{pilot}},\ldots,\vec{Y}_{L}^{\mathsf{pilot}}) \reqdef U$, and using the same beamformers up to a possibly random scaling factor $\vec{p}^{\mathsf{dl}} \in \set{P}_U^K$, and the same instantaneous sum power $KP$.

We remark that assuming availability of the received uplink pilots at the decoders is operationally meaningless, in the downlink. However, this assumption still provides an effective approximation of the optimistic ergodic rates in \eqref{eq:oer}, which themselves serve as accurate approximations of performance under practical downlink pilot signaling schemes. Another important aspect is that the condition $\vec{p}^{\mathsf{dl}} \in \set{P}_U^K$ implies that the downlink precoders achieving the desired rate tuple may also be functions of $U$. This further implies that the rates shown in Figure~\ref{fig:fixed_sum_coh} may not be achievable under the original information constraint (e.g., in the cell-free model with no CSI sharing). However, they are rigorously achievable if the access points are allowed to coordinate by sharing $K$ instantaneous power scaling coefficients, in addition to the originally assumed CSI sharing pattern for beamforming. Given these conditions, the rate profiles in Figure~\ref{fig:fixed_sum_coh} correspond to the ergodic rates induced by (almost surely) optimal points on the (Pareto) boundary of the instantaneous rate region, under an instantaneous sum power constraint of $KP$. Finally, we note that similar reasoning adapted to the case of ideal decoder CSI can be used to interpret the rates in Figure~\ref{fig:fixed_sum_oer} as downlink optimistic ergodic rates.

\begin{figure}[htp]
    \centering
    \begin{subfigure}{0.9\columnwidth}
        \centering
    \includegraphics[width=\textwidth]{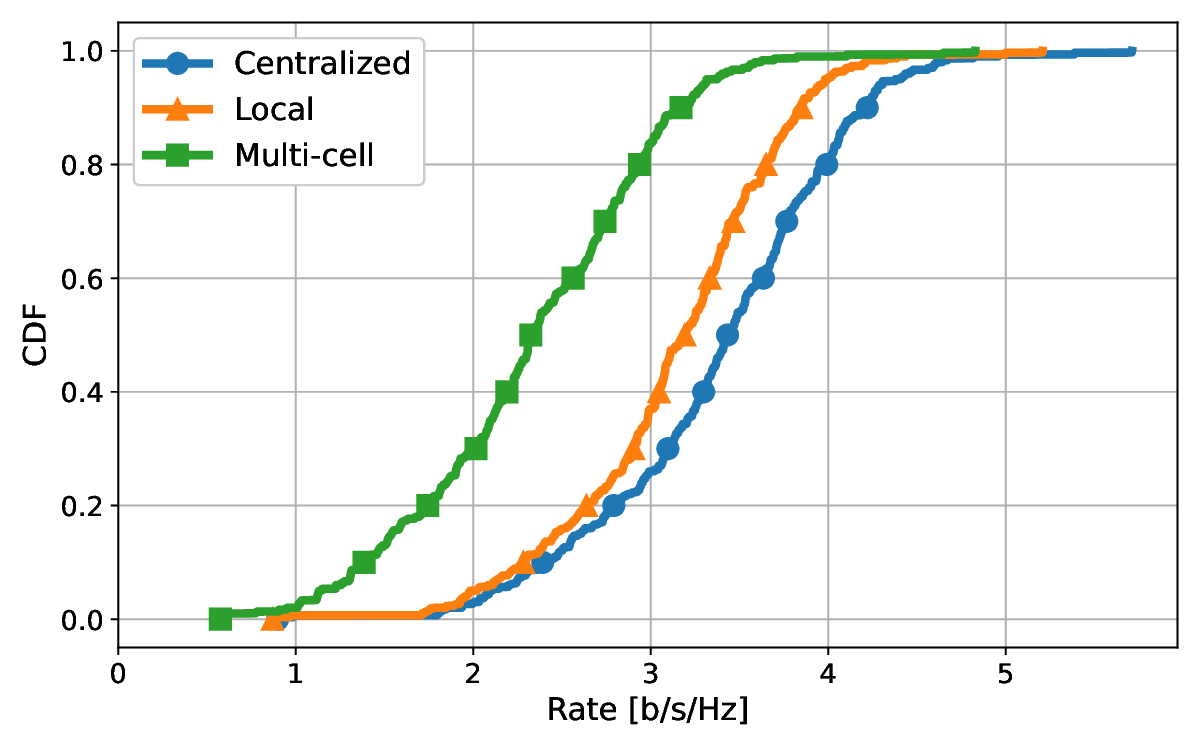}
        \caption{}
        \label{fig:uatf_uatf}
    \end{subfigure}
    \begin{subfigure}{0.9\columnwidth}
        \centering
    \includegraphics[width=\textwidth]{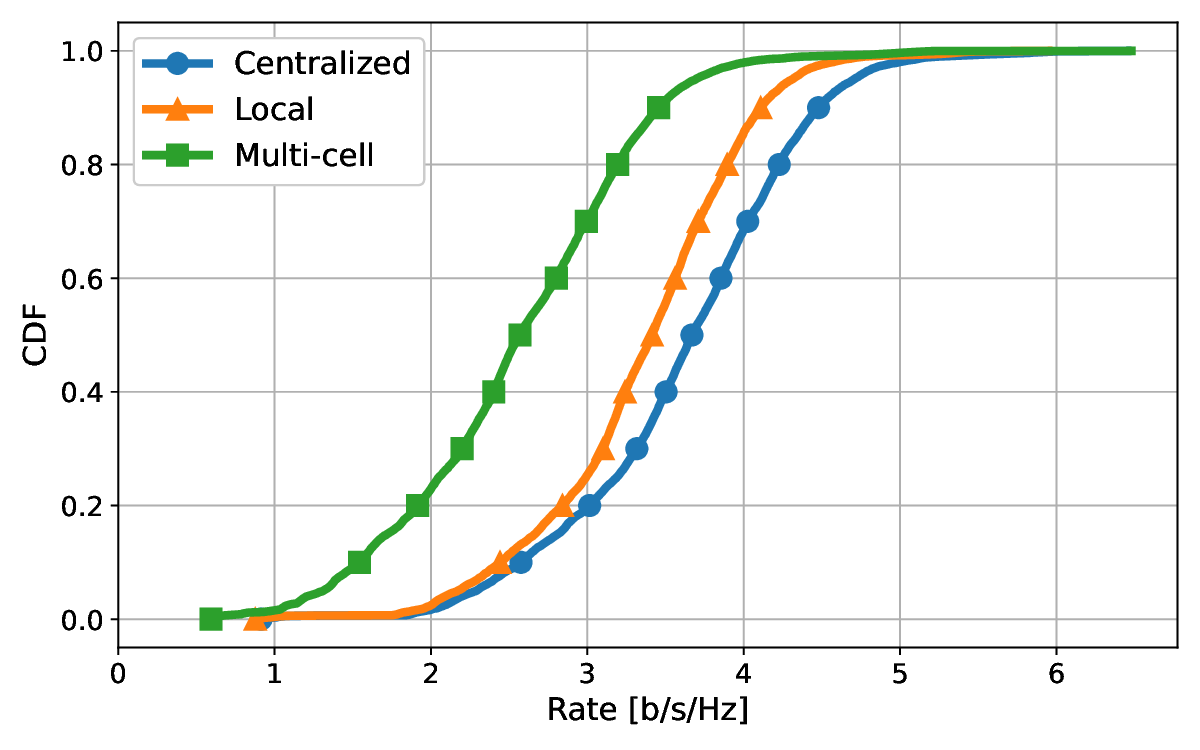}
        \caption{}
        \label{fig:uatf_oer}
    \end{subfigure}
    \caption{Empirical CDF of uplink ergodic rates achieved by solutions to \eqref{prob:maxmin_uatf} for different beamforming architectures. The  ergodic rates are approximated by using (a) the UatF lower bound in \eqref{eq:uatf}, and (b) the optimistic upper bound in \eqref{eq:oer}.}
    \label{fig:uatf}
\end{figure}

\subsection{Jointly optimal power control}
Figure~\ref{fig:uatf} reports the performance of an optimal solution to the two-timescale max-min uplink joint beamforming, channel estimation, and power control problem \eqref{prob:maxmin_uatf}, assuming a per-user power constraint of $P = 20$ dBm and considering the same information constraints as in the previous section. We first remark that, by construction, Figure~\ref{fig:uatf_uatf} provides optimal rate profiles in terms of worst-case UatF rates. Indeed, the lower tails show better worst-case performance than the heuristic max-min policy used for Figure~\ref{fig:fixed_uatf}, at the expense of the higher rates. Interestingly, the same trend can also be observed for the upper bound in Figure~\ref{fig:uatf_oer}, which suggests once more that the UatF bound (or equivalently, the MSE criterion) may serve as a good proxy for optimizing less tractable rate bounds.

While the coherent decoding lower bound is less amenable to long-term resource allocation, it may, by contrast, be well-suited for short-term resource allocation, i.e., for optimization problems that target Pareto-optimal points on the instantaneous achievable rate region for each channel realization. For this to hold, one important condition is that the decoder CSI is also available as common beamforming CSI. As an example, Figure~\ref{fig:coh} reports the performance of the short-term uplink max-min joint beamforming, channel estimation, and power control problem \eqref{prob:maxmin_coh}. The local case is not considered, as it does not satisfy the common beamforming CSI condition. A key feature of the solutions to \eqref{prob:maxmin_coh} is that the worst-case instantaneous rate is maximized for every channel realization. However, note that this differs from maximizing the worst-case rate on average, i.e., the worst-case ergodic rate. We can observe that the lower tails of \eqref{prob:maxmin_coh} are similar (in a max-min sense) to the lower tails of Figure~\ref{fig:uatf_oer}, despite the higher complexity associated with short-term optimization. We refer to \cite{miretti2023fixed,miretti2024spawc} for additional discussions on this aspect.

We note that similar considerations apply to downlink versions of Figure~\ref{fig:uatf} and Figure~\ref{fig:coh}, which are therefore omitted.

\begin{figure}[htp]
    \centering    \includegraphics[width=0.9\columnwidth]{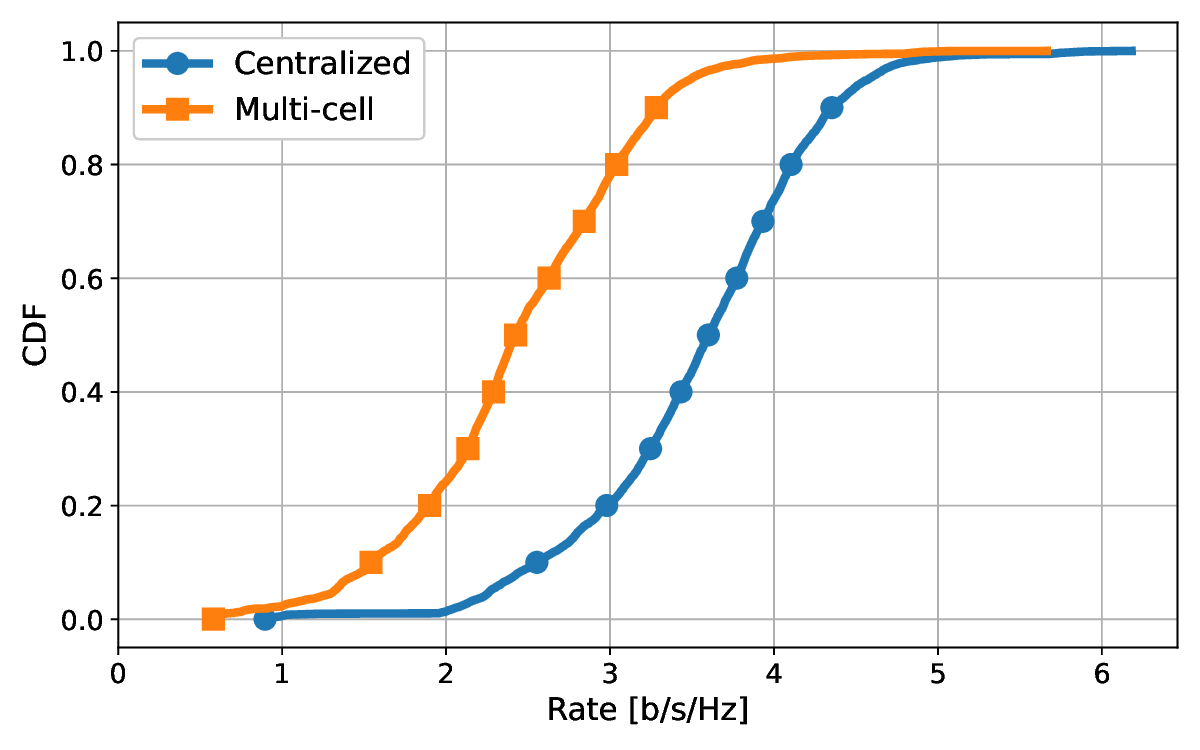}
    \caption{Empirical CDF of uplink optimistic ergodic rates \eqref{eq:oer} achieved by solutions to \eqref{prob:maxmin_coh} for different beamforming architectures.}
    \label{fig:coh}
\end{figure}

\section{Conclusion}
We have established that, under broad and practically relevant conditions, joint channel estimation and beamforming in massive MIMO systems can be optimally separated into MMSE estimation followed by MMSE beamforming. These conditions encompass both centralized and distributed architectures and align with standard system models and ergodic rate bounds in the literature. The results are particularly strong for the UatF/hardening bound and extend partially to the coherent decoding lower bound under additional assumptions regarding the availability of common beamforming CSI. This separation principle enables complete solutions to specific problems (e.g., max-min resource allocation) and reveals strong structural properties for more challenging problems (e.g., sum-rate maximization), which are reduced to the tuning of power control coefficients.

Our findings support the continued use of modular channel estimation and beamforming designs in systems where the considered models accurately reflects reality. Although illustrated under a block-fading measurement model, we remark that the results extend naturally to more explicit measurement models for wideband channel estimation and prediction (e.g., based on 5G sounding reference signals), provided that the measurements remain linear and Gaussian. In all such cases, joint design offers limited additional performance gains, and artificial intelligence techniques may be better directed at reducing the complexity of intractable resource allocation tasks, such as finding sum-rate optimal power control coefficients.

However, in scenarios involving non-linearities, non-Gaussian effects, unknown channel statistics, or alternative performance metrics (e.g., outage capacity or block error rate), the separation principle may no longer apply. In these cases, joint designs, possibly driven by artificial intelligence methods, may still offer significant performance gains.

\appendix
\subsection{Proof of Proposition~\ref{prop:coh}}\label{app:coh}
The following steps follow from standard information inequalities \cite{el2011network}:
\begin{align*}
R_k & = I(x_k,y_k|U_k) \\
	& = h(x_k)-h(x_k|y_k,U_k) \\
	& = \sup_{\alpha \in \set{F}_{U_k}}h(x_k)-h(x_k-\alpha y_k|y_k,U_k) \\
	& \geq \sup_{\alpha \in \set{F}_{U_k}} h(x_k)-h(x_k-\alpha y_k|U_k) \\
	&\geq \sup_{\alpha \in \set{F}_{U_k}} h(x_k)-\E\left[\log\left(\pi e \E\left[|x_k-\alpha y_k|^2|U_k\right]\right)\right] \\
	&= \sup_{\alpha \in \set{F}_{U_k}} \log(\pi e)-\E\left[\log\left(\pi e \E\left[|x_k-\alpha y_k|^2|U_k\right]\right)\right] \\
	&= -\inf_{\alpha \in \set{F}_{U_k}}\E\left[\log\left(\E\left[|x_k-\alpha y_k|^2|U_k\right]\right)\right].
\end{align*} 
The third equality follows from the translation invariance property of the differential entropy $h(x_k|y_k,U_k) = h(x_k-\alpha y_k|y_k,U_k)$, which holds for all functions $\alpha$ of $U_k$, and by taking the supremum over the set of functions $\set{F}_{U_k}$. The two inequalities follow since conditioning reduces entropy, and from the maximum differential entropy lemma, respectively. 

By the monotonicity of the logarithm, and by definition of \textit{essential} infimum (see notation section), we then have
\begin{align}\label{eq:essinf_alpha}
&\inf_{\alpha \in \set{F}_{U_k}}\E\left[\log\left(\E\left[|x_k-\alpha y_k|^2|U_k\right]\right)\right] \geq \E[\log(\varepsilon(U_k))],
\end{align}
where $\varepsilon(U_k) \eqdef \essinf{\alpha \in \set{F}_{U_k}}\E\left[|x_k-\alpha y_k|^2|U_k\right]$.
Since $\alpha$ is a function of $U_k$, the random variable $\varepsilon$ can be characterized by minimizing the quadratic form $\E\left[|x_k-\alpha(u_k) y_k|^2|U_k=u_k\right]= \E[|x_k|^2]
+\E[|y_k|^2|U_k=u_k]|\alpha(u_k)|^2 -2\Re(\E[x_k^* y_k|U_k=u_k]\alpha(u_k)) $ over $\alpha(u_k)\in \stdset{C}$ for each realization $u_k \in \mathcal{U}_k$ of $U_k$. More precisely, we notice that $(\forall u_k \in \set{U}_k)$ 
\begin{equation*}
\alpha^\star(u_k) \eqdef \begin{cases}\frac{\E[x_k y_k^*|U_k=u_k]}{\E[|y_k|^2|U_k=u_k]} & \text{if } \E[|y_k|^2|U_k=u_k]\neq 0,  \\
0 & \text{otherwise}, \end{cases}
\end{equation*}
attains the pointwise infimum
\begin{equation*}
\varepsilon(u_k) = \begin{cases}1-\frac{|\E[x_k^* y_k|U_k=u_k]|^2}{\E[|y_k|^2|U_k=u_k]} & \text{if } \E[|y_k|^2|U_k=u_k]\neq 0,  \\
1 & \text{otherwise.} \end{cases}
\end{equation*}
Since $\alpha^\star \in \set{F}_{U_k}$, it attains the essential infimum, and hence the inequality in \eqref{eq:essinf_alpha} is an equality.\footnote{Using an indicator function $\mathbbm{1}:\set{U}_k \to \{0,1\}$ for the subset $\set{U}_k^0 \eqdef \{u_k\in \set{U}_k~:~\E[|y_k|^2|U_k = u_k]=0\}$, we notice that $(\forall u_k\in \set{U}_k)$ $\alpha(u_k) = \frac{\E[x_k y_k^*|U_k=u_k]}{\E[|y_k|^2|U_k=u_k]+\mathbbm{1}(u_k)}$, since $\E[|y_k|^2|U_k=u_k]=0\implies \E[x_k^\star y_k|U_k=u_k]=0$.  The indicator function avoids divisions by zero, and it is measurable because $\set{U}_k^0$ is a measurable set. Then, $\alpha$ is measurable since it involves sum, product, and (well-defined) quotient of measurable functions.} 
The proof is completed by noting that $\E[x_k^\star y_k|U_k]=\E[g_{k,k}|U_k]$, $\E[|y_k|^2|U_k]=\sum_{j\in \set{K}}\E[|g_{j,k}|^2|U_k]+\E[\sigma_k^2|U_k]$, and by rearranging the terms in $R_k \geq -\E[\log(\varepsilon)]$.

\subsection{Proof of Proposition~\ref{prop:uatf}}\label{app:uatf}
By proof of Proposition~\ref{prop:coh} (Appendix~\ref{app:coh}), we recall that $R_k^{\mathsf{coh}}$ corresponds to the left hand size of $\eqref{eq:essinf_alpha}$. We then obtain
\begin{align*}
R_k^{\mathsf{coh}} &= -\inf_{\alpha \in \set{F}_{U_k}} \E\left[\log\left(\E\left[|x_k-\alpha y_k|^2|U_k\right]\right)\right]\\
&\geq -\inf_{\alpha \in \set{F}_{U_k}} \log\left(\E\left[|x_k-\alpha y_k|^2\right]\right)\\
&\geq -\inf_{\alpha \in \stdset{C}} \log\left(\E\left[|x_k-\alpha y_k|^2\right]\right)\\
&= -\log\left(\inf_{\alpha \in \stdset{C}}\E\left[|x_k-\alpha y_k|^2\right]\right) = R_k^{\mathsf{UatF/hard}},
\end{align*}
where the first inequality follows from Jensen's inequality and the law of total expectation, the second inequality follows since constant functions are a subset of $\set{F}_{U_k}$, and where the last two equalities follow from the monotonicity of the logarithm and  by minimizing the convex quadratic form $\E\left[|x_k-\alpha y_k|^2\right]$ over $\alpha\in \stdset{C}$, as in the proof of Proposition~\ref{prop:coh}.

\subsection{Proof of Proposition~\ref{prop:stationarity}}\label{proof:stationarity}
Following the same arguments as in the proof of \cite{miretti2021team}[Lemma~6], we know that if $\vec{v}^\star_k$ attains the infimum, then it also solves Problem~\eqref{eq:stationarity} (and attains a finite MSE, since $0\leq \mathsf{MSE}_k(\vec{v}_k^\star)\leq \mathsf{MSE}_k(\vec{0}) = 1$). As an informal explanation of these arguments, we notice that the equations in \eqref{eq:stationarity} can be obtained by fixing $\vec{v}_{j,k}$ for all $j\neq l$ and by minimizing the MSE pointwise over $\vec{v}_{l,k}$ for each realization of $S_{l,k}$. This readily gives a set of necessary optimality conditions, reminiscent of the game theoretical notion of \textit{Nash equilibrium}. Note that all the expectations in \eqref{eq:stationarity} are well-defined (i.e., finite) since $\E[\|\vec{H}\vec{H}^\herm\|_\mathrm{F}^2]< \infty$, $\vec{p}\in \set{P}^K_{S_{l,k}}$, and $\mathsf{MSE}_k(\vec{v}_k)<\infty \implies \E[\|\vec{v}_k\|^2]<\infty$. For the converse statement, we notice that if $\vec{v}_k^\star$ solves Problem~\eqref{eq:stationarity} and attains a finite MSE, then it satisfies the sufficient optimality conditions in \cite[Theorem~2.6.4]{yukselbook}, and hence it attains the infimum. Furthermore, \cite[Theorem~2.6.4]{yukselbook} also states that the infimum is attained by a unique $\vec{v}_k^\star$. 

\subsection{Proof of Proposition~\ref{prop:ess}}\label{proof:ess}
We start from the equivalent characterization of $\vec{v}_k^\star$ in Proposition~\ref{prop:stationarity}. Specifically, if $\vec{v}_k^\star \in \set{V}_k$ attains $\inf_{\vec{v}_k\in \set{V}_k}\mathsf{MSE}_k(\vec{v}_k)$, then it is the unique solution to Problem~\eqref{eq:stationarity}. Hence, we can express each subvector $\vec{v}_{l,k}^\star$ of $\vec{v}_k^\star$ as
\begin{multline*}
\vec{v}_{l,k}^\star =\Big(\E[\vec{H}_l\vec{P}\vec{H}_l^\herm|S_{l,k}',U_k]+\vec{I}_N\Big)^{-1}\\
\Big(\E[\vec{H}_l|S_{l,k}',U_k]\vec{P}^{\frac{1}{2}}\vec{e}_k-\sum_{j\in \set{L}_k\backslash \{l\}}\E[\vec{H}_l\vec{P}\vec{H}_j^\herm\vec{v}_{j,k}^\star|S_{l,k}',U_k] \Big).
\end{multline*}
We notice that $\vec{v}_k^\star$ can be rewritten as $\vec{v}_{k}^\star=\vec{v}_{k}^{(U_k)}$, where we define $(\forall u_k\in \set{U}_k)$ the beamformers $\vec{v}_k^{(u_k)}$ with subvectors
\begin{multline*}
\vec{v}_{l,k}^{(u_k)} \eqdef \Big(\E[\vec{H}_l\vec{P}\vec{H}_l^\herm|S_{l,k}',U_k=u_k]+\vec{I}_N\Big)^{-1}\\
\Big(\E[\vec{H}_l|S_{l,k}',U_k=u_k]\vec{P}^{\frac{1}{2}}\vec{e}_k\\
-\sum_{j\in \set{L}_k\backslash \{l\}}\E[\vec{H}_l\vec{P}\vec{H}_j^\herm\vec{v}_{j,k}^{(u_k)}|S_{l,k}',U_k=u_k] \Big).
\end{multline*}
We further notice that, for every $u_k \in \set{U}_k$ such that $\E[\|\vec{H}\vec{H}^\herm\|_\mathrm{F}^2|U_k=u_k]< \infty$, $\vec{v}_k^{(u_k)}$ solves a modified version of Problem~\eqref{eq:stationarity} obtained by replacing $\E[\cdot]$ with $\E_{u_k}[\cdot]\eqdef \E[\cdot|U_k=u_k]$ everywhere, and hence it satisfies $\inf_{\vec{v}_k\in \set{V}_k}\mathsf{MSE}_k(\vec{v}_k|U_k=u_k)=\mathsf{MSE}_k(\vec{v}_k^{(u_k)}|U_k=u_k)$. Since $\E[\|\vec{H}\vec{H}^\herm\|_\mathrm{F}^2]< \infty$ implies $\E[\|\vec{H}\vec{H}^\herm\|_\mathrm{F}^2|U_k]< \infty$, we then obtain $\essinf{\vec{v}_k\in \set{V}_k}\mathsf{MSE}_k(\vec{v}_k|U_k) \geq \mathsf{MSE}_k(\vec{v}_k^{(U_k)}|U_k) = \mathsf{MSE}_k(\vec{v}_k^\star|U_k)$. Therefore, $\vec{v}_k^\star$ attains the essential infimum.

\subsection{Proof of Proposition~\ref{prop:local}}
\label{app:local}
The considered channel and CSI model satisfies the assumptions of Proposition~\ref{prop:distr}. Substituting \eqref{eq:LTMMSE} into the  the optimality conditions \eqref{eq:TMMSE} for a given UE $k\in\set{K}$, we obtain $(\forall l \in \set{L}_k)$
\begin{align*}
\vec{v}_{l,k} &=\vec{V}_{l,k}\left(\vec{e}_k-\sum_{j \in \set{L}_k\backslash  \{l\}} \vec{P}^{\frac{1}{2}}\E\left[\hat{\vec{H}}_{j}^{(k)\herm}\vec{V}_{j,k}\vec{c}_{j,k}\Big|S_{l,k}\right] \right) \\
&= \vec{V}_{l,k}\left(\vec{e}_{k} - \sum_{j \in \set{L}_k\backslash  \{l\}} \vec{\Pi}_{j,k}\vec{c}_{j,k} \right) = \vec{V}_{l,k} \vec{c}_{l,k},
\end{align*}
where the second equality follows from the independence between $S_{l,k}$ and $(\hat{\vec{H}}_j^{(k)},\vec{V}_{j,k})$ for $j\neq l$, and where the last equality follows since $(\forall l \in \set{L}_k)~\vec{c}_{l,k} + \sum_{j \in \set{L}_k\backslash  \{l\}} \vec{\Pi}_j\vec{c}_{j,k} = \vec{e}_k$ holds by construction. This last system of equations always admits a unique solution, as proven next.

We focus for simplicity on the case $\set{L}_k= \set{L}$, since the extension to arbitrary $\set{L}_k$ is immediate. We rewrite the system as $(\vec{M}_k+\vec{U}\vec{\Pi}_k) \vec{c}_k= \vec{U}\vec{e}_k$, where $\vec{c}_k:= \begin{bmatrix}
\vec{c}_{1,k}^\T & \ldots &\vec{c}_{L,k}^\T 
\end{bmatrix}^\T$, $\vec{\Pi}_k:= \begin{bmatrix}
\vec{\Pi}_{1,k} & \ldots &\vec{\Pi}_{L,k} 
\end{bmatrix}$, $\vec{U}:= \begin{bmatrix}
\vec{I}_K & \ldots &\vec{I}_K 
\end{bmatrix}^\T $, $\vec{M}_k:=\mathrm{diag}(\vec{I}_K-\vec{\Pi}_{1,k},\ldots,\vec{I}_K-\vec{\Pi}_{L,k})$. We shall prove that $\vec{M}_k+\vec{U}\vec{\Pi}_k$ is invertible. By the matrix inversion lemma (see, e.g., \cite[Appendix~C]{boyd2004convex}), $\vec{M}_k+\vec{U}\vec{\Pi}_k$ is invertible if both $\vec{M}_k$ and $\vec{I}_K+\vec{\Pi}_k\vec{M}^{-1}_k\vec{U}$ are invertible. Furthermore, standard arguments show that $(\forall l \in \set{L})~\vec{0}\preceq \vec{P}^{\frac{1}{2}}\hat{\vec{H}}_l^{(k)\herm}\vec{V}_{l,k} \preceq \vec{P}^{\frac{1}{2}}\hat{\vec{H}}_l^{(k)\herm}(\hat{\vec{H}}_l^{(k)}\vec{P}\hat{\vec{H}}_l^{(k)\herm} + \vec{I}_N)^{-1}\hat{\vec{H}}_l^{(k)}\vec{P}^{\frac{1}{2}} \prec \vec{I}_K$ holds almost surely, where we use $\preceq$ ($\prec$) to denote the partial ordering (strict partial ordering) with respect to the cone of Hermitian positive semidefine matrices as in \cite{boyd2004convex} (Loewner order). Therefore, $(\forall l \in \set{L})~\vec{0}\preceq \vec{\Pi}_{l,k} \prec \vec{I}_K$ holds. This in turn implies that $\vec{M}_k$ and $\vec{I}_K+\vec{\Pi}_k\vec{M}^{-1}_k\vec{U} = \vec{I}_K + \sum_{l\in \set{L}}\vec{\Pi}_{l,k}(\vec{I}-\vec{\Pi}_{l,k})^{-1}$ are Hermitian positive definite, hence invertible.
(The eigenvalues of $\vec{\Pi}_{l,k}(\vec{I}-\vec{\Pi}_{l,k})^{-1}$ take the form $\xi/(1-\xi)\geq 0$, where $0\leq\xi<1$ is an eigenvalue of $\vec{\Pi}_{l,k}$.)

\subsection{Proof of Proposition~\ref{prop:duality_hard}}\label{proof:duality_hard}
We first focus on the case $(\gamma_1,\ldots,\gamma_K)\in \stdset{R}_{++}^K$. If $(\vec{v}_1^{\star},\ldots,\vec{v}_K^{\star},\vec{p}^\mathsf{ul})$ solves Problem~\eqref{prob:feas_UatF}, then $\vec{D} > \vec{0}$ and hence $\vec{\Sigma}>\vec{0}$ must hold. Moreover,
by rearranging the uplink SINR constraints, we have that $\vec{p}^{\mathsf{ul}} = \vec{D}^{-1}\vec{\Gamma}\vec{B}^\T\vec{p}^{\mathsf{ul}}+\vec{D}^{-1}\vec{\Gamma}\vec{\Sigma}\vec{1}$, i.e., $\vec{p}^{\mathsf{ul}}$ is a fixed-point of the positive affine mapping $T^{\mathsf{ul}}: \stdset{R}_{+}^K \to \stdset{R}_{++}^K : \vec{p} \mapsto  \vec{D}^{-1}\vec{\Gamma}\vec{B}^\T\vec{p}+\vec{D}^{-1}\vec{\Gamma}\vec{\Sigma}\vec{1}$. From the properties of positive affine mappings, which are standard interference mappings \cite{yates95} \cite[Lemma 1]{cavalcante2019connections}, we know that the set of fixed-points of $T^{\mathsf{ul}}$ is either a singleton or the empty set. Hence, $\vec{p}^{\mathsf{ul}} = (\vec{D}\vec{\Gamma}^{-1}-\vec{B}^\T)^{-1}\vec{\Sigma}\vec{1}$. We further know that a fixed-point exists if and only if the spectral radius $\rho(\vec{D}^{-1}\vec{\Gamma}\vec{B}^\T)$ of $\vec{D}^{-1}\vec{\Gamma}\vec{B}^\T$ satisfies $\rho(\vec{D}^{-1}\vec{\Gamma}\vec{B}^\T)<1$ \cite[Ch. 2]{stanczak2009fundamentals} \cite[Prop. 4]{cavalcante2019connections}. Using basic linear algebra, we then have that $\rho(\vec{D}^{-1}\vec{\Gamma}\vec{B}^\T)= \rho(\vec{D}^{-1}\vec{\Gamma}\vec{B})<1$. The latter inequality implies the existence and uniqueness of $\vec{p}^{\mathsf{dl}}\in \stdset{R}_{++}^K$ such that $\vec{p}^{\mathsf{dl}} = \vec{D}^{-1}\vec{\Gamma}\vec{B}\vec{p}^{\mathsf{dl}}+\vec{D}^{-1}\vec{\Gamma}\vec{1}$, and hence that $(\vec{v}_1^\star,\ldots,\vec{v}_K^\star,\vec{p}^{\mathsf{dl}})$ with $\vec{p}^{\mathsf{dl}} = (\vec{D}\vec{\Gamma}^{-1}-\vec{B})^{-1}\vec{1}$ solves Problem~\eqref{prob:feas_hard}. We then have the direct calculation
\begin{align*}
\vec{1}^{\T}\vec{\Sigma}\vec{p}^\mathsf{dl} &= \vec{1}^{\T}\vec{\Sigma}(\vec{D}\vec{\Gamma}^{-1}-\vec{B})^{-1}\vec{1} \\
&= \left(\vec{1}^{\T}\vec{\Sigma}(\vec{D}\vec{\Gamma}^{-1}-\vec{B})^{-1}\vec{1}\right)^\T \\
&= \vec{1}^{\T}(\vec{D}\vec{\Gamma}^{-1}-\vec{B}^\T)^{-1}\vec{\Sigma}\vec{1} = \vec{1}^{\T}\vec{p}^\mathsf{ul}.
\end{align*}
The more general case $(\gamma_1,\ldots,\gamma_K)\in \stdset{R}_{+}^K$ follows by applying the same arguments as above to reduced dimension feasibility problems obtained by removing all trivial constraints $\gamma_k=0$ and corresponding variables $(\vec{v}_k,p_k)$. The converse statement follows by reversing all the above arguments.

\subsection{Proof of Proposition~\ref{prop:hard_pertx_aug}}\label{app:proof_hard_pertx_aug}
We first give an alternative version of \cite[Proposition~2]{miretti2024duality} and \cite[Proposition~3]{miretti2024duality} that does not require strict feasibility of Problem~\eqref{eq:hard_pertx}. In what follows we tacitly assume that $\set{V}_k$ is restricted to elements with finite second-order moment, to exploit Hilbert space arguments. We remark that this incurs no loss of generality, since the objective of Problem~\eqref{eq:hard_pertx} is finite by assumption.

\begin{lemma}\label{lem:partial_dual}
Given the subset $\set{V}_{\vec{\gamma}} \eqdef \{(\vec{v}_1,\ldots,\vec{v}_K) \in \set{V}_1\times \ldots \times \set{V}_K ~|~ (\forall k \in \set{K})~\mathsf{SINR}_k^{\mathsf{hard}}(\vec{v}_1,\ldots,\vec{v}_K) \geq \gamma_k\}$ of beamformers satisfying the SINR constraints in Problem~\eqref{eq:hard_pertx}, define the partial dual function $(\forall \vec{\lambda}\in \stdset{R}_+^L)$
\begin{equation*}
 d(\vec{\lambda}) \eqdef \inf_{\substack{(\forall k \in \set{K})\\ \vec{v}_k\in \set{V}_k}} \sum_{k=1}^K \E[\|\vec{v}_k\|_{\vec{1}+\vec{\lambda}}^2] +\delta(\vec{v}_1,\ldots,\vec{v}_K) -\sum_{l=1}^L \lambda_l P_l,
\end{equation*}
where $\delta(\vec{v}_1,\ldots,\vec{v}_K) = 0$ if $(\vec{v}_1,\ldots,\vec{v}_K) \in \set{V}_{\vec{\gamma}}$, and $\delta(\vec{v}_1,\ldots,\vec{v}_K) = +\infty$ otherwise. Strong duality holds, i.e., $\sup_{\vec{\lambda}\in \stdset{R}_{+}^L}d(\vec{\lambda})=p^\star$, where $p^\star$ denotes the optimum of Problem~\eqref{eq:hard_pertx}. Furthermore, there exist Lagrangian multipliers $\vec{\lambda}^\star$ such that $d(\vec{\lambda}^\star)=p^\star$. Moreover, Problem~\eqref{eq:hard_pertx} admits a solution $(\vec{v}_1^\star,\ldots,\vec{v}_K^\star)$.
\end{lemma}
\begin{proof}
For all $(\vec{v}_1,\ldots,\vec{v}_K) \in \set{V}_1\times \ldots \times \set{V}_K$ and $k \in \set{K}$, we define the functions
 $f_k(\vec{v}_1,\ldots,\vec{v}_K) \eqdef \sqrt{\textstyle \sum_{j=1}^K\E[|\vec{h}_k^\herm\vec{v}_j|^2]+1} - \nu_k \Re\left(\E[\vec{h}_k^\herm\vec{v}_k]\right)$, where $\nu_k\eqdef \sqrt{1+1/\gamma_k}$, obtained by rearranging the SINR constraints and by replacing $|\E[\vec{h}_k^\herm\vec{v}_k]|$ with $\Re\left(\E[\vec{h}_k^\herm\vec{v}_k]\right)$. Note that for all $(\vec{v}_1,\ldots,\vec{v}_K) \in \set{V}_1\times \ldots \times \set{V}_K$ and $k \in \set{K}$ $f_k(\vec{v}_1,\ldots,\vec{v}_K) \leq 0 \implies \mathsf{SINR}_k^{\mathsf{hard}}(\vec{v}_1,\ldots,\vec{v}_K) \geq \gamma_k$, due to the simple property $(\forall a \in \stdset{C})~\Re(a) \leq |a|$. The functions $f_1,\ldots,f_K$ are continuous convex functions (see, e.g., proof of \cite[Lemma~1]{miretti2024duality}). We use these functions to construct the closed convex subset $\set{V}_{\vec{\gamma}}' \eqdef \{(\vec{v}_1,\ldots,\vec{v}_K) \in \set{V}_1\times \ldots \times \set{V}_K ~|~ (\forall k \in \set{K})~f_k(\vec{v}_1,\ldots,\vec{v}_K) \leq 0\}$ of $\set{V}_{\vec{\gamma}}$. We then consider the following convex version of Problem~\eqref{eq:hard_pertx} 
\begin{equation}\label{prob:indicator}
\begin{aligned}
\underset{(\forall k\in \set{K})~\vec{v}_k\in \set{V}_k} {\text{minimize}} \quad &  \sum_{k=1}^K\E[\|\vec{v}_{k}\|^2] + \delta'(\vec{v}_1,\ldots,\vec{v}_K)\\
\text{subject to } \; \quad & (\forall l \in \set{L})~\sum_{k=1}^K\E[\|\vec{v}_{l,k}\|^2]\leq P_l,
\end{aligned}
\end{equation}
where $\delta'(\vec{v}_1,\ldots,\vec{v}_K) = 0$ if $(\vec{v}_1,\ldots,\vec{v}_K)\in \set{V}_{\vec{\gamma}}'$, and $\delta'(\vec{v}_1,\ldots,\vec{v}_K) = +\infty$ otherwise. We remark that $\delta'$ is a proper convex function, since it is the indicator function of a closed convex set \cite{bauschke2011convex}. We also consider the alternative dual function $(\forall \vec{\lambda}\in \stdset{R}_+^L)$
\begin{equation*}
d'(\vec{\lambda}) \eqdef \inf_{\substack{(\forall k \in \set{K})\\ \vec{v}_k\in \set{V}_k}} \sum_{k=1}^K \E[\|\vec{v}_k\|_{\vec{1}+\vec{\lambda}}^2] +\delta'(\vec{v}_1,\ldots,\vec{v}_K) -\sum_{l=1}^L \lambda_l P_l.
\end{equation*}
Denote by $p'^\star$ the optimum of Problem \eqref{prob:indicator}. Since the objective and all constraints of Problem \eqref{prob:indicator} are proper convex functions, and since Slater's condition is trivially satisfied due to the nonzero power constraints, strong duality holds \cite{bauschke2011convex}\cite[Proposition~1]{miretti2024duality}, i.e., $\sup_{\vec{\lambda}\in \stdset{R}_{+}^L}d'(\vec{\lambda})=p'^\star$. Furthermore, there exists $\vec{\lambda}^\star\in \stdset{R}_{+}^L$ such that $d'(\vec{\lambda}^\star)=\sup_{\vec{\lambda}\in \stdset{R}_{+}^L}d'(\vec{\lambda})$ \cite{bauschke2011convex}\cite[Proposition~1]{miretti2024duality}. 

We then readily have $(\forall \vec{\lambda} \in \stdset{R}_+^L)$ $d(\vec{\lambda}) \leq p^\star \leq p'^\star = d'(\vec{\lambda}^\star)$,
where the first inequality follows from weak duality \cite{bauschke2011convex}\cite[Proposition~1]{miretti2024duality} and by embedding the SINR constraints in the objective function via the indicator function $\delta$ of $\set{V}_{\vec{\gamma}}$ as in Problem~\eqref{prob:indicator}, and where the second inequality follows since $\set{V}_{\vec{\gamma}}'\subseteq \set{V}_{\vec{\gamma}}$. Following the same arguments as in \cite[Lemma~4]{miretti2024duality}, it can be finally shown that $(\forall \vec{\lambda}\in \stdset{R}_+^L) ~d'(\vec{\lambda}) = d(\vec{\lambda})$, and hence that $d(\vec{\lambda}^\star) = p^\star = p'^\star = d'(\vec{\lambda}^\star)$.

For the last part of the statement, we remark that feasibility of Problem~\eqref{eq:hard_pertx} implies feasibility of Problem~\eqref{prob:indicator}. Then, by interpreting Problem~\eqref{prob:indicator} as a projection onto a nonempty closed convex set, existence and uniqueness of a solution $(\vec{v}_1',\ldots,\vec{v}_K')$ to Problem~\eqref{prob:indicator}  follows by the Hilbert projection theorem, as in the proof of \cite[Lemma~5]{miretti2024duality}. Since $p^\star = p'^\star$, then $(\vec{v}_1',\ldots,\vec{v}_K')$ is also a solution to Problem~\eqref{eq:hard_pertx}.
\end{proof}

It remains to show that a primal solution can be recovered from a dual solution $\vec{\lambda}^\star$ as defined in Lemma~\ref{lem:partial_dual}. This can be done by following the same arguments as in the proof of \cite[Proposition~4]{miretti2024duality}, which do not require strict feasibility.


\bibliographystyle{IEEEbib}
\bibliography{IEEEabrv,refs}

\end{document}